\documentclass[journal]{IEEEtran}
\usepackage{textcomp}
\usepackage{amsfonts,amsmath,amssymb}
 \usepackage{amsthm}
\usepackage{graphicx}
\usepackage{verbatim}
\usepackage{float}
\usepackage{cite}
\usepackage{subfigure}
\usepackage{cases}
\usepackage{enumerate}
\usepackage{algpseudocode}
\usepackage{xcolor}
\newcommand \remove[1] {}
\newcommand \hide[1] {}

\def \ham {\mathcal{H}}
\def \P {\mathbb{P}}

\newtheorem{Theorem}{Theorem}
\newtheorem{Lemma}{Lemma}

\newtheorem{Property}{Property}
\newcommand{\bF}{\mathbf{F}}
\newcommand{\bFu}[1]{\bF^{(#1)}}
\newcommand{\boldf}{\mathbf{f}}
\newcommand{\bI}{\mathbf{I}}
\newcommand{\BS}{\mathbb{S}}
\newcommand{\bS}{\mathbf{S}}
\newcommand{\bu}{\mathbf{u}}
\newcommand{\bx}{\mathbf{x}}
\newcommand{\by}{\mathbf{y}}
\newcommand{\taus}{s}

\begin{document}
\clubpenalty=10000
\widowpenalty = 10000

\floatstyle{ruled}
\newfloat{algorithm}{thp}{lop}
\floatname{algorithm}{Algorithm}

\pagenumbering{arabic}
\title{Optimal Energy-Aware Epidemic Routing in DTNs}

\author{
Soheil Eshghi,
MHR. Khouzani,
Saswati Sarkar,
Ness B. Shroff,
Santosh S. Venkatesh%
}
\maketitle
{\renewcommand{\thefootnote}{} \footnotetext{
S. Eshghi, S. Sarkar and S. S. Venkatesh are with the Department of Electrical and Systems Engineering at the University of Pennsylvania,
Philadelphia, PA. Their email addresses are \emph{eshghi,swati,venkates@seas.upenn.edu}, their mailing address is: 200 S. 33rd Street,
Philadelphia, PA 19104, tel.~: (215) 573-9071, and fax : (215) 573-2068. MHR. Khouzani and N. B. Shroff are with the Department of Electrical and Computer Engineering at the Ohio State University, Columbus, OH. Their e-mail addresses are \emph{khouzani,shroff@ece.osu.edu}, their mailing address is: 2015 Neil Avenue, Columbus, OH 43210, tel.~: (614) 247-6554, and fax: (614) 292-7596.\ \\
This paper was presented [in part] at the ACM International Symposium on Mobile Ad Hoc Networking and Computing (MobiHoc '12), Hilton Head Island, SC, June 2012.\ \\This work is partially supported by the Army Research Office MURI Awards W911NF-08-1-0238 and W911NF-07-1-0376, the Defense Thrust Research Agency (DTRA) grant HDTRA1-14-1-0058, and NSF grants CNS-0831919, CNS-0721434, CNS-1115547, CNS-0915697, CNS-0915203, and CNS-0914955.}}
\begin{abstract}
In this work, we investigate the use of epidemic routing in energy constrained Delay Tolerant Networks (DTNs).
In epidemic routing, messages are relayed by intermediate nodes at contact opportunities, i.e., when pairs of nodes come within the transmission range of each other. Each node needs to decide whether to forward its message upon contact with a new node based on its own residual energy level and the age of that message.
We mathematically characterize the fundamental trade-off between energy conservation and a measure of Quality of Service
as a dynamic energy-dependent optimal control problem. We prove 
that in the mean-field regime, the optimal dynamic forwarding decisions follow simple threshold-based structures in which the forwarding threshold for each node depends on its current remaining energy.  We then characterize the nature of this dependence\hide{dependence of these thresholds on  current energy reserves in each node}. Our simulations reveal that the optimal dynamic policy significantly outperforms heuristics.
\end{abstract}
\hide{\terms{Theory}}
\begin{IEEEkeywords}
DTN, 
Energy-Based Epidemic Routing, Stratified Optimal Control,  Threshold-Based Forwarding.
\end{IEEEkeywords}
\IEEEpeerreviewmaketitle

\section{Introduction}

Delay Tolerant Networks (DTNs) have been envisioned for\hide{disaster response and military/tactical networks} civilian disaster response networks, military networks, and environmental surveillance, e.g., where communication devices are carried by disaster relief personnel and soldiers, or where they can be mounted on wandering animals. These networks are comprised of mobile nodes whose communication range is much smaller than their roaming area, and therefore messages are typically relayed by intermediate nodes at times of spatial proximity.\hide{ Opportunistic networking based on DTNs is envisioned to assist ad hoc or cellular communication in next generation networks whenever end-to-end connectivity is hard to achieve \cite{repantis2005data}.} 
Relaying messages consumes a signicant amount of energy in the sending and receiving nodes. However, mobile nodes in DTNs typically have limited  battery reserves and replacing/recharging the batteries of drained nodes is usually
infeasible or expensive\hide{not cost-efficient}. Simple epidemic forwarding depletes the limited energy reserves of nodes, while conservative forwarding policies jeopardize the timely delivery of the message to the destination.
\hide{Simple epidemic forwarding
schemes that rely on flooding as a method of message propagation may deplete the energy reserves of the mobile nodes, thereby undermining the performance of the network.
A depleted node, or one with critically low battery reserve, will not be able to 
relay new messages in the future.
This reduction in the number of relay nodes will in turn reduce the long term throughput of the network, specially for time sensitive messages.
On the other hand, an overly conservative packet forwarding strategy would compromise the timely delivery of the message to the destination.}Hence, there is an inherent trade-off between
timely message delivery and energy conservation.  

The literature on message routing in DTNs is extensive~\cite{vahdat2000epidemic, zhang2007performance,lindgren2003probabilistic,de2009nectar,banerjee2010design,spyropoulos2005spray,wang2006dft,lu2010energy,nelson2009encounter,de2010optimal,singh2011optimal,singh2010delay,balasubramanian2007dtn,altman2010optimal,neglia2006optimal}. {Most notably, Vahdat and Becker \cite{vahdat2000epidemic} present a policy where each node propagates the message to all of its neighbours simultaneously (``Epidemic Routing''), while Spyropoulos {\em et al.}~\cite{spyropoulos2005spray} propose spreading a specific number of copies of the message initially and then waiting for the recipients of these copies
to deliver the message \hide{directly }to the destination (``Spray and Wait''). Wang and Wu~\cite{wang2006dft} present {`` Optimized Flooding''}, where flooding is stopped once the total probability of message delivery exceeds a threshold. Singh \emph{et al.}~\cite{singh2011optimal} and Altman \emph{et al.}~\cite{altman2010optimal} \hide{respectively }identify optimal \hide{(in terms of total energy)} and approximately optimal message forwarding policies in the class of policies that do not take the distribution of node energies into account.} In summary, the state of the art in packet forwarding in DTNs comprises of heuristics that ignore energy constraints\cite{vahdat2000epidemic,zhang2007performance,lindgren2003probabilistic}, those that consider only overall energy consumption but provide no analytic performance guarantees\cite{wang2006dft,de2009nectar,banerjee2010design,spyropoulos2005spray,lu2010energy,nelson2009encounter}, and those that do not utilize the energy available {\em to each node} in making forwarding decisions\cite{de2010optimal,singh2011optimal,singh2010delay,balasubramanian2007dtn,altman2010optimal,neglia2006optimal} (we describe some of these policies in more detail in \S\ref{sec:Numericals}).
{An efficient forwarding strategy can use knowledge of the distribution of energy among among nodes to its advantage, and this motivates the design of dynamic energy-dependent {controls} which are the subject of this paper.}

 We start by formulating the trade-off between energy conservation and likelihood of timely delivery
 as a dynamic energy-dependent optimal control problem: at any given time, each node {chooses} its forwarding probability based on its current remaining energy. Since the number of relay nodes with the message increases and residual energy reserves decrease with transmissions and receptions,
 the forwarding probabilities vary with time. 
Thus, they must be chosen so as to control the evolution of network states, which capture both the fraction of nodes holding a copy of the message and the  remaining battery reserves of the nodes.
We model the evolution of \hide{network}these states using  epidemiological differential equations that rely on mean-field approximation of Markov processes, and seek dynamic forwarding probabilities ({\em optimal controls}) that optimize objective functions penalizing  energy depletion subject to enforcing timely message delivery (\S\S\ref{subsec:model_dynamics},\ref{subsec:objective}).  \hide{These dynamic forwarding probabilities constitute our optimal controls.}

Our first result is to prove that dynamic optimal controls follow simple 
threshold-based rules (\S\ref{sec:Structural_Results}, Theorem~\ref{Thm:Gen_structure}). That is, a node in possession of a copy of the message forwards the message to nodes it encounters that have not yet received it until a
certain threshold time that depends on its current remaining energy. Calculating these thresholds is much simpler than solving the general problem and can be done once at the source node of the message. Subsequently, they can be added to the message as a small overhead. Each node that receives the message can retrieve the threshold times and forward the message if its age is less than the threshold entry of the node's residual energy level. The execution of the policy at each node is therefore simple  and based only on local information.{\hide{As the node forwards the message, it loses energy,
and its {time threshold} changes accordingly.  If the age of the message is past {the} threshold time corresponding to a node's current level of energy, it stops forwarding the message to others and will only transmit a copy to the destination of the message. The simplicity of the analytically-derived strategies
is somewhat surprising given that the system dynamics involve non-linear {terms} and a vector of controls.}}{\hide{ This result and its proof constitute important contributions of this paper.}}\hide{The proofs for optimality
of threshold-type policies in such cases do not follow from existing optimal control results.}

Our second result is to characterize the nature of the dependence of the thresholds on the energy levels.
Intuitively, the less energy a node has, the more reluctant it should be to transmit the message, as the transmission will drive it closer to critically low battery levels{\hide{ (which in turn will impair timely delivery of future messages)}}.
However, our investigations reveal that this intuition
can only be confirmed when the penalties associated with low final remaining energies are convex (\S\ref{sec:Structural_Results}, Theorem~\ref{Thm:Order}), and does not hold in general otherwise.
{\hide{That is, in the former case, higher remaining energy levels lead to longer forwarding durations, but
 the monotonicity of the thresholds in energy levels is not necessarily preserved otherwise.}}

{Finally, our optimal control provides a missing \emph{benchmark} for forwarding policies in large networks in which no information about the mobility pattern of the individual nodes is available and a minimum QoS is desired. This benchmark
 allows us to observe the sub-optimality of some simpler heuristic policies, and to identify parameter ranges in which they perform close to the optimal\hide{identify  some  even simpler heuristic policies that perform close to the optimal control, and
 also those that substantially compromise performance for simplicity} (\S\ref{sec:Numericals}).

\section{System Model}
We assume a low-load scenario in which only one message is propagated in the network within a terminal time $T$. \hide{We present the model for a \emph{single-delivery} setting. Particularly, the}This message has a single destination and it is sufficient for a copy of the message to be delivered to its destination by the terminal time.\hide{\footnote{\hide{The single-delivery}This setting also
captures cases in which there are multiple destinations that are, unlike the source and the intermediate nodes,  inter-connected through a backbone network, \hide{An example of such a setting is where}e.g., when the final destination is (a group of) base stations in a cellular network. Note that in \hide{the latter}this case there is no additional benefit in delivering more than one copy of the message to one of the destinations, since once one destination receives the message, it can instantly inform the other destinations of its reception using the high-speed backbone network. So we can assume all destination nodes receive the message virtually simultaneously. The only modification in our modeling would be a re-scaling of the rate of contact between a mobile node and the destination by the number of these inter-connected sinks (Base Stations).}} We use the deterministic mean-field (i.e., for large numbers of nodes) regime which models state evolution using a system of differential equations. Such models have been shown to be acceptable approximations both analytically and empirically for large and fast-moving mobile wireless networks~\cite{khouzani2012optimal}. \hide{In~\cite{techreport}, we independently investigate the accuracy of this model using simulations.  }In \S\ref{subsec:model_dynamics}, we develop our system dynamics model based on mean-field deterministic ODEs. 
Subsequently, in \S\ref{subsec:objective} we consider two classes of utility functions that cogently combine a penalty for the impact of
the policy on the residual energy of the nodes with guarantees for the QoS of the forwarding policy.

\subsection{System Dynamics}\label{subsec:model_dynamics}
We begin with some definitions: a node that has received a copy of the message and is not
its destination  is referred to as an \emph{infective}; a (non-destination) node that has not yet received a copy of the message is called a \emph{susceptible}.
The maximum energy capacity of all nodes is $B$ units.
 A message transmission between a pair of nodes consumes $\taus$ units of energy in the transmitter and $r$ units in the receiver{, independent of their total energy level.} 
Naturally,  $r\leq\taus$.
When an infective node contacts a susceptible at time $t$, the message is transmitted with a certain forwarding probability if the infective (transmitter) and susceptible (receiver) have at least $\taus$  and $r$
units of energy. {If either does not have the respective sufficient energy, transmission will not occur.}

Two nodes contact each other at rate $\hat\beta$.  We assume that inter-contact times are exponentially distributed and uniform among nodes, an assumption common to many mobility models (e.g., Random Walker, Random Waypoint, Random Direction, etc. \cite{groenevelt2005message}). Moreover, it is shown in~\cite{groenevelt2005message} that
\begin{equation}\label{eq:hat_beta}
\hat\beta \propto \frac{\text{average rel. speed of nodes}\times\text{communication ranges}}{\text{the roaming area}}.
\end{equation}
Assuming $t=0$ mark the moment of message generation, we define $S_i(t)$   (respectively, $I_i(t)$) to be the \emph{fraction} of nodes that are susceptible (respectively, infective) and that have $i$ energy units at time $t$.
Hence for $t\in [0,T]$: $\sum_{i=0}^{B} \left(S_i(t)+I_i(t)\right)=1$. 

At any given time, each node can observe its own
level of available energy, and its forwarding decision should, in general, utilize such information.
Hence, upon an instance of contact between a susceptible node with $i$ units of energy and an infective node with $j$ units of energy at time $t$, as long as $i\geq r$ \textbf{and} $j\geq\taus$, the message is {forwarded} with probability $u_j(t)$ ($0\leq u_j(t)\leq 1$). { We take these probabilities to be our controls $\bu(t) = \bigl(u_\taus(t),u_{\taus+1}(t),\dots,u_B(t)\bigr)\in \mathcal{U}$, where $\mathcal{U}$ is the set of piecewise continuous controls with left-hand limits\hide{ existing} at each $t \in (0, T]$, and right-hand limits\hide{ existing} at each $t \in [0, T)$.} If the message is forwarded},
the susceptible node transforms to an infective node with ${i-r}$ energy units, and the infective node {likewise} to an infective node with ${j-\taus}$ energy units.
{We assume that once an infective contacts another node, the infective can identify (through a low-load exchange of control messages) whether the other node has a copy of the message (i.e., is infective), or does not (i.e., is susceptible), whether the contacted node is a destination and also whether it has enough energy to receive the message.
We assume that the dominant mode of energy consumption is the transmission and reception of the message, and that each exchange of the control messages consumes an insignificant amount of energy.} {If a message-carrying node that has sufficient energy for one transmission contacts the destination that has yet to receive the message, the message is always forwarded to the destination.}

Let $N$ be the total number of nodes and define $\beta:=N\hat\beta$.
Following~\eqref{eq:hat_beta}, $\hat\beta$ is inversely proportional to the roaming area, which scales with $N$.
Hence, if we can define a density of nodes, $\beta$ has a nontrivial value.
The system dynamics in \hide{the mean-field}this regime \hide{(i.e., for large $N$) }over any finite interval can be approximated thus, { except at the finite points of discontinuity of $\bu$} \hide{ as follows }(\!\!\cite[Theorem~1]{gast2011mean}):

\begin{subequations}\label{model:no_recharge}
\begin{numcases}{\dot{S}_i =}
-\beta S_i\sum_{j=s}^B u_j I_j ~~\quad\quad\qquad\qquad (r\leq i\leq B),
    \label{no_rech:S_gen}\\
 0 ~\qquad\qquad\qquad\qquad\qquad\qquad (0\leq i<r), \label{no_rech:S_margin} 
\end{numcases}
\begin{numcases}{\dot{I}_i =}
 - \beta u_i I_i \sum_{j = r}^B S_j  \qquad\qquad\qquad~(B-r< i \leq B),
    \label{no_rech:I_hi_margin}\\ 
 -\beta u_i I_{i} \sum_{j=r}^{B}S_j
    + \beta S_{i+r} \sum_{j=\taus}^{B}u_j I_j~~\notag
       \\\qquad\qquad\qquad\qquad\qquad~~(B-\taus<i\leq B-r), \label{no_rech:I_midhi_margin}\\ 
 -\beta u_i I_{i} \sum_{j=r}^{B}S_j
    + \beta S_{i+r} \sum_{j=\taus}^{B}u_j I_j  + \beta u_{i+s} I_{i+s} \sum_{j=r}^B S_j\notag
      \\\qquad\qquad\qquad\qquad\qquad\qquad~~(\taus\leq i\leq B-\taus), \label{no_rech:I_gen}\\ 
   \beta S_{i+r} \sum_{j=\taus}^{B}u_j I_j
      + \beta u_{i+\taus}I_{i+\taus} \sum_{j=r}^{B}S_j
      ~ (0\leq i<\taus). \label{no_rech:I_low_margin} 
\end{numcases}
\end{subequations}
 Note that in the above differential equations and in the rest of the paper, whenever not ambiguous, the dependence on $t$ is made implicit.
We now explain each of these equations:\footnote{We will consider protocols where a destination receives at most one copy of the message by the terminal time. Nothing changes in the system dynamics if we allow multiple copies to be received because isolated transmissions have no effect on the mean-field regime.\hide{The system dynamics can ignore the single instance of the delivery of the message to the destination. This is because {when the state is represented as the fraction of nodes of each type in the mean-field regime, i.e. for large $N$}, the change in the energy distributions as a result of a single transmission of the message is negligible.
Note that in the single-delivery scenario, once the destination receives the message, subsequent contacts between infectives and the destination will not result in a transmission of the message.}
}\\
\eqref{no_rech:S_gen}: The rate of decrease in the fraction of susceptible nodes with energy level $i\geq r$ is proportional to the rate of contacts between
these nodes and {transmitting} infective nodes with energy level equal to or higher than $\taus$.\\
\eqref{no_rech:S_margin}: Susceptibles with less than $r$ units of energy cannot convert to infectives.\\
{ The rate of change in infectives of energy level $i$ is due to three mechanisms:

\begin{enumerate}
\item {Transmitting} infectives of energy level $i$ convert {to} infectives with energy level $i-\taus$ upon contact with susceptibles that have sufficient energy for message exchange. This conversion happens due to the energy consumed in transmitting the message, resulting in a decrease in infectives of energy level $i$.\label{case1}

\item Susceptibles with energy level  $i+r$  are transformed {to} infectives of energy level $i$ upon contact with {transmitting} infectives that have at least $\taus$ units of energy, swelling the ranks of infectives of energy level $i$. This conversion occurs due to the energy consumed in receiving the message.

\item {Transmitting} infectives of energy level $i+\taus$ convert {to} infectives with energy level $i$ upon contact with susceptibles that have sufficient energy for message exchange, adding to the pool of infectives of energy level $i$. Like \ref{case1}, this is due to the energy consumed in transmitting the message.
\end{enumerate}
Now, given that energy levels are upper-bounded by $B$:

\begin{enumerate}[I]
\item If $B-r<i\leq B$, only mechanism 1 is possible, as $i+\taus\geq i+r > B$, ruling out 2 and 3 respectively. This results in~\eqref{no_rech:I_hi_margin}.

\item If $B-\taus<i\leq B-r$, only mechanisms 1 and 2 are possible, as $i+\taus>B$ rules out 3, leading to~\eqref{no_rech:I_midhi_margin}. 

\item If $\taus\leq i\leq B-\taus$, all three mechanisms are in play, resulting in~\eqref{no_rech:I_gen}.

\item If $0\leq i< \taus$, only mechanisms 2 \& 3 remain, as $i-\taus<0$ rules out 1. Thus, we have~\eqref{no_rech:I_midhi_margin}.
\end{enumerate}
}

\hide{\eqref{no_rech:I_hi_margin}: The rate of decrease in the fraction of infective nodes with energy level $i$ such that $i>B-r$ is proportional to their rate of contact with any susceptible with more than $r$ units of energy. {\bf{No susceptible or infective node can become an infective with this high level of energy by receiving or sending the message, as the act of message-passing will deplete its reserves too much.}}\\
\eqref{no_rech:I_midhi_margin}: Similarly, the rate of change in the fraction of infectives with energy level $i$ such that \hide{between }$B-\taus < i \leq B-r$\hide{, along with the mechanism in~\eqref{no_rech:I_hi_margin},} is due to the transformation of susceptibles with energy level  $i+r$  upon contact with {transmitting} infectives that have at least $\taus$ units of energy, along with the mechanism in~\eqref{no_rech:I_hi_margin}. No infective 
can be transformed to an infective of such a high energy level.\\
\eqref{no_rech:I_gen}: This is the non-marginal equation for the evolution of the infectives. Here, three mechanisms are in place: \eqref{no_rech:I_hi_margin}, \eqref{no_rech:I_midhi_margin}, and one more: {transmitting} infectives of energy level $i+\taus$ convert
to infectives with energy level $i$ upon contact with susceptibles that have sufficient energy for message exchange.\\
\eqref{no_rech:I_low_margin}: Infectives with less than $\taus$ units of energy cannot convert to any other type.}

{We consider { continuous} state solutions} $\bS(t)= \bigl(S_0(t),  \dots, S_B(t)\bigr)$, $\bI(t) = \bigl(I_0(t),  \dots, I_B(t)\bigr)$ to the dynamical system~\eqref{model:no_recharge} subject to initial conditions
\begin{equation}\label{state:initial}
\begin{split}
\mathbf{S}(0) &=\mathbf{S}_0:=(S_{00},\ldots,S_{0B}),\\
\mathbf{I}(0) &=\mathbf{I}_0:=(I_{00},\ldots,I_{0B}).
\end{split}
\end{equation}
 We naturally assume that the initial conditions satisfy $\bS(0)\geq\mathbf{0}$, $\bI(0)\geq\mathbf{0}$, and $\sum_{i=0}^B \bigl(S_i(0) + I_i(0)\bigr) = 1$ (vector inequalities are to be interpreted component-wise throughout). 

We say that a state solution $(\bS(t)$, $\bI(t))$ for the system~\eqref{model:no_recharge} is \emph{admissible} if the non-negativity and {normalization} conditions
\begin{align}\label{states:constraints}
 \mathbf{S}(t)\geq\mathbf{0}, \quad \mathbf{I}(t)\geq\mathbf{0}, \quad
 \sum_{i=0}^B\left(S_{i}(t)+I_i(t)\right)=1,
\end{align}
are satisfied for all $t\in[0,T]$\hide{Here and elsewhere,}. 
We next show that {\hide{a solution of}states satisfying}~\eqref{model:no_recharge} are admissible and unique for any $\bu~\in~\mathcal{U}$:
\begin{Theorem}\label{thm:constraints}
Suppose the initial conditions satisfy $\bS(0)\geq\mathbf{0}$, $\bI(0)\geq\mathbf{0}$, and $\sum_{i=0}^B \bigl(S_i(0) + I_i(0)\bigr) = 1$, and suppose $\bu(t) = \bigl(u_\taus(t),u_{\taus+1}(t),\dots,u_B(t)\bigr)$ is any system of piecewise continuous controls. Then the dynamical system~\eqref{model:no_recharge} has a unique state solution $(\bS(t)$, $\bI(t))$, which is admissible. If $I_i(t') > 0$ for any $i$ (respectively, $S_j(t')>0$ for any $j$) and $t' \in [0, T),$ $I_i(t) > 0$ (respectively $S_j(t)>0$) for all $t > t'$. Also, for each $j,$  $S_j(t') = 0$ for all $t' \in (0, T]$ if $S_j(0) = 0.$
\end{Theorem}
 
In our proof, we use the following  general result: \hide{(proved in Appendix-\ref{appendix_lemma_1}):}
\begin{Lemma}
\label{baselemma}
Suppose the vector-valued function $\boldf = (f_i,1\leq i\leq N)$  has component functions given by quadratic forms $
  f_i(t, \bx) = \bx^T Q_i(t) \bx \quad (t\in[0,T];\; \bx\in\BS)$,
  where $\BS$  {is the set of $N$-dimensional vectors} $\bx = (x_1,\dots, x_N)$ satisfying $\bx\geq\mathbf{0}$ and $\sum_{i=1}^N x_i = 1$, and $Q_i(t)$ is a  matrix whose  components are uniformly, absolutely bounded over $[0,T]$. Then, for an $N$-dimensional vector-valued function $\bF$, the system of differential equations 
\begin{equation}\label{Vectordf}
\begin{split}
  \dot{\bF}(t) = \boldf(t, \bF) \qquad(0<t\leq T)\\
    \quad\text{subject to initial conditions $\bF(0)\in \BS$}
\end{split}
\end{equation}
 has a unique solution,
 $\bF(t)$, which varies continuously with the initial conditions $\bF_0 \in \BS$  at each $t\in[0,T]$.
 \end{Lemma}

{This follows from standard results in the theory of ordinary differential equations~\cite[Theorem A.8, p. 419]{seierstad1987optimal}  given the observation that $\boldf(t,\bF)$ is comprised of quadratic forms and is thus Lipschitz over $[0,T]*\BS$.}

\begin{proof}{
We write $\bF(0) = \bF_0$,\hide{, in a slightly informal notation, we may thus} and in a slightly informal notation, $\bF = \bF(t) = \bF(t, \bF_0)$ to acknowledge the dependence of $\bF$ on the initial value $\bF_0$. }

We first verify the normalization condition of the admissibility criterion.  
By summing the left and right sides of the system of equations~\eqref{model:no_recharge} we see that
$
  \sum_{i=0}^B\bigl(\dot{S}_i(t) + \dot{I}_i(t)\bigr) = 0,
$
and, in view of the initial normalization 
{$\sum_{i=0}^B\bigl(S_i(0) + I_i(0)\bigr) = 1$, we have 
$ \sum_{i=0}^B\bigl(S_i(t) + I_i(t)\bigr) = 1$ for all $t$.} 

We now verify the non-negativity condition.  Let $\bF = (F_1,\dots,F_N)$ be the  state vector in $N = {2(B+1)}$ dimensions whose elements are comprised of $(S_i,0\leq i\leq B)$ and $(I_i, 0\leq i\leq B)$ in some order.\hide{,  and  $\BS$  {is the set of $N$-dimensional vectors} $\bx = (x_1,\dots, x_N)$ satisfying $\bx\geq\mathbf{0}$ and {$\sum_{i=1}^N x_i = 1$}.} The system of equations~\eqref{no_rech:S_gen}--\eqref{no_rech:I_low_margin} can thus be represented as $\dot{\bF} = \boldf(t, \bF)$, where for $t\in[0,T]$ and  $\bx\in\BS$, the vector-valued function $\boldf = (f_i,1\leq i\leq N)$ has component functions  $
  f_i(t, \bx) = \bx^T Q_i(t) \bx
$
in which $Q_i(t)$ is a  matrix whose non-zero elements are of the form $\pm\beta u_k(t)$. Thus,  the components of $Q_i(t)$ are uniformly, absolutely bounded over $[0,T]$.\hide{ We may now write the system~(\ref{no_rech:S_gen}--\ref{no_rech:I_low_margin}) compactly in the form of the vector differential equation~\eqref{Vectordf},\hide{
} the solutions of which depend on time and also implicitly on the initial conditions. }\hide{Writing  $\bF(0) = \bF_0$, in a slightly informal notation, we may acknowledge this dependence and write $\bF = \bF(t) = \bF(t, \bF_0)$. } Lemma~\ref{baselemma} establishes that the solution $\bF(t, \bF_0)$ to the system
 \eqref{no_rech:S_gen}--\eqref{no_rech:I_low_margin} is unique and varies continuously with the initial conditions $\bF_0$; it clearly varies continuously  with time.   Next, using elementary calculus, we show in the next paragraph that if $\bF_0\in \text{\bf Int }\BS$ (and, in particular, each component of $\bF_0$ is positive), then each component of the solution $\bF(t, \bF_0)$ of \eqref{no_rech:S_gen}--\eqref{no_rech:I_low_margin}  is positive at each $t\in[0,T]$.\footnote{Throughout the paper, we use positive for strictly positive, etc.}  Since  $\bF(t, \bF_0)$  varies continuously with $\bF_0$, it follows that $\bF(t,\bF_0)\geq\mathbf{0}$ for all $t\in[0,T]$, {$\bF_0\in \BS$}, which completes the overall proof.

Accordingly, let each component of $\bF_0$ be positive. Since the solution $\bF(t, \bF_0)$  varies continuously with time,
there exists a time, say $t' > 0$, such that each component of  $\bF(t, \bF_0)$ is positive
in the interval $[0, t').$ The result follows trivially if $t' \geq T$. Suppose now that there exists $t''<T$
such that each component of  $\bF(t, \bF_0)$ is positive
in the interval $[0, t'')$, and at least one component is $0$ at $t''.$
We first show that such components can not be $S_i$ for any $i \geq 0$ and subsequently
rule out $I_i$ for all $i \geq 0.$  Note that $u_j(t), I_j(t), S_j(t)$ are bounded in $[0, t'']$ (recall $\sum_{j=0}^B\left(S_j(t)+I_j(t)\right) = 1, S_j(t)
\geq 0, I_j(t) \geq 0$ for all $j, t \in [0, t'']$). First, let $r \leq i \leq B.$ From \eqref{no_rech:S_gen},
$S_i(t'') = S_i(0) e^{-\beta\int_{0}^{t''} \sum_{j=s}^B u_j(t)I_j(t) \, dt}.$
Since all $u_j(t), I_j(t)$ are bounded in $[0, t'']$,  and $S_i(0) > 0$, $\beta > 0$, therefore $S_i(t'') > 0$.
From \eqref{no_rech:S_margin}, $S_i(t'') = S_i(0) > 0$ for $0 \leq i < r.$ Thus, $S_i(t'') > 0$ for all $i.$ Since $S_i(t) > 0$, $I_i(t) \geq 0$ for all $i, t \in [0, t'']$, from \eqref{no_rech:I_hi_margin} --
\eqref{no_rech:I_gen}, $\dot{I}_i \geq -\beta u_i I_i \sum_{j=r}^B S_j$  for all $i \geq \taus$ in the interval $[0,t'']$. Thus, $I_i(t'') \geq I_i(0) e^{-\beta\int_{0}^{t''} u_i(t)\sum_{j=r}^B S_j(t) \, dt}.$ Since all $u_j(t), I_j(t), S_j(t)$ are bounded in $[0, t'']$, and $I_i(0) > 0, \beta > 0$,  it follows that $I_i(t'') > 0$ for all $i \geq \taus.$ Finally, since $S_i(t) > 0, I_i(t) \geq 0$ for all $i, t \in [0, t'']$,  from \eqref{no_rech:I_low_margin}, it  follows that $\dot{I}_i \geq 0$ for all $i < \taus$, $t \in [0, t'']$. Thus, $I_i(t'') \geq I_i(0) > 0$
for all $i < \taus.$ This contradicts the definition of $t''$ and in turn implies that $\bF(t,\bF
_0)>0$ for all $t \in [0,T]$, $\bF_0 \in \text{\bf Int } \BS$.

Since the control and the unique state solution $\bS(t)$, $\bI(t)$ are non-negative, \eqref{no_rech:S_gen}
implies that  $\bS(t)$ is a non-increasing function of time. Thus, $S_j(t) = 0$ if $S_j(0) = 0$ for any $j.$  Using the  argument in the above paragraph and starting from a $t' \in [0, T)$ where $S_j(t') > 0$, or $I_j(t') > 0$,
it may be shown that $S_j(t) > 0$ or $I_j(t) > 0$ respectively for all $t > t'.$
 \end{proof}


The above proof allows for choices of $T$
that depend on the controls $\bu$, provided such controls result in finite $T$. For the problem to be non-trivial, we assume henceforth that there exist $i\geq r$, $j\geq \taus$ for which $S_{i}(0)>0$ and $I_{j}(0)>0$.

{ We conclude this section with another technical lemma which we will use later:

\begin{Lemma}\label{dotI}
For all $t\in (0,T)$ and all $i$, $|\dot{I}_i(t^+)|$ and $|\dot{I}_i(t^-)|$ exist and are bounded, as is $|\dot{I}_i(T^-)|$.
\end{Lemma}
\begin{proof}
 The states are admissible~(Theorem \ref{thm:constraints}) and continuous, and the controls are bounded by definition. Hence, due to~\eqref{model:no_recharge}, $|\dot{I}_i (t)|$ exists and is bounded at all points except the finite set of points of discontinuity of the controls, and continuous over each interval over which $\bu$ is continuous. Thus, $|\dot{I}_i(t^+)|$ and$|\dot{I}_i(t^-)|$ exist and are bounded for all $t\in (0,T)$. Using the same reasoning,  $|\dot{I}_i(T^-)|$ also exists and is bounded.
\end{proof}}

\vspace{-0.2in}
\subsection{{Throughput constraint and }objective functions}\label{subsec:objective}
The objective function of the network {can represent} both a measure of the efficacy of the policy in ensuring timely message delivery, and the effect of the policy on the residual energy reserves of the nodes. \hide{In what follows, }{We first develop measures for each of these cases, and then utilize them to define an objective function and a constraint on the achieved network throughput.}

\paragraph*{Throughput constraint}
One plausible measure of QoS in the context of DTNs is the probability of delivery of the message to the destination \hide{(in the single-delivery scenario) or to as many destination nodes as possible (in the multi-delivery scenario)}before a terminal time $T$. 
We examine two cases: one in which a minimum probability of delivery is mandated on the message before a fixed terminal time $T$, and another in which the time-frame of message delivery is flexible and the goal is to meet the minimum probability of delivery requirement as soon as possible. In what follows, we discuss these two cases.



Let $\hat\beta_0$ be the rate of contact of a node with the destination,  potentially different from $\hat\beta$,
and define $\beta_0:=N\hat\beta_0$.

Following from the exponential distribution of the inter-contact times, the {{\em mandated probability of delivery}} constraint 
$\mathbb{P}(\text{delivery})\geq p$ (i.e., the message being delivered to the destination with probability greater than or equal to $p$ within $[0,T]$) implies that:
$
1-\exp\left(-\int_{0}^{T}\!\beta_0\sum_{i=\taus}^B I_i(t)\,dt\right)\geq p
$. \footnote{This is because $P(\text{delivery})=\mathbb{E}\{\mathbf{1}_{\sigma=t}\}$, where $\sigma$ is the time of delivery of the message to the destination. Therefore:
\begin{align*}
P(delivery)&=\hide{\mathbb{E}\{\mathbf{1}_{\sigma=t}\}
=}\int_{0}^T\!\ P(\sigma=t)\,dt\\
 &=\int_{0}^T\! \exp \left(-\hat\beta_0\int_{0}^{t}\!\sum_{i=\taus}^B N I_i(\xi)\,d\xi \right)
\cdot \hat\beta_0 \sum_{i=\taus}^BNI_i(t)\,dt\\&=1-\exp\left(-\int_{0}^{T}\!\beta_0\sum_{i=\taus}^B I_i(t)\,dt\right).
\end{align*}
A special case of this \hide{formula} was shown in~\cite[Appendix A]{groenevelt2005message}
and \cite[\S II.A]{altman2010optimal}.}

 Note that the {\hide{second}exponential} term in the LHS is the probability that no contact occurs between the destination and  any infective with sufficient energy  during the interval of $[0,T]$.

Also notice that similar to~\eqref{eq:hat_beta}, $\hat\beta_0$ is inversely proportional to the roaming area, which itself scales with $N$.
{\hide{Therefore, as long as we can define a meaningful density (number of nodes divided by their total roaming area), 
this probability is nontrivial.}} 
Another point to note is that the summation inside the integral starts from index $\taus$, since infective nodes with less than $\taus$ units of energy cannot forward their message to the destination upon potential contact.
This is equivalent to a throughput constraint:
\begin{align}\label{path_constraint}
 \int_{0}^{T}\!\sum_{i=\taus}^B I_i(t)\,dt\geq -\ln (1-p)/\beta_0.
\end{align}

In the first case, referred to as the \emph{fixed terminal time} problem,  the terminal time $T$ is fixed and the throughput constraint is satisfied
(along with minimizing the adverse effects on the residual energy of the nodes which we will discuss next) {through appropriate choice of control function $\bu$, if any such functions exist.}
In the second case, referred to as the \emph{optimal stopping time} problem,  for every  choice of the  control function $\bu$, the  terminal time $T$ is chosen to
satisfy \eqref{path_constraint}{\hide{, specifically}} with equality. The terminal time
is therefore variable and depends on the choice of $\bu.$ {\hide{There exists}Such a $T$ exists} for a given control $\bu$ if and only if for the resulting states
\begin{equation}{ \label{omit}
\lim_{T'\rightarrow\infty} \int_{0}^{T'}\!\sum_{i=\taus}^B I_i(t)\,dt \geq  -\ln (1-p)/\beta_0.}
\end{equation}
 The throughput constraint will not be satisfied in any finite time horizon for controls
 that do not satisfy the above. We will therefore exclude such controls in the optimizations we formulate next. Note that if the system uses a zero-control
  (i.e., $\bu(t) = (0, \ldots, 0)$ at all $t$) then $S_i(t) = S_i(0)$ and $I_i(t) = I_i(0)$ for all $t$; thus, since $\sum_{i=s}^B I_i(0) > 0$, \eqref{omit} holds. Therefore, there exists at least one control
  that satisfies \eqref{omit}.  Since $T$ is finite for every control that satisfies \eqref{omit},
  the system is admissible for each such control as well.

\paragraph*{Energy cost of the policy}
In the simplest representation of the trade-off with the energy overhead,
one can think of maximizing the aggregate remaining energy in the network at the terminal time, irrespective of how it is distributed. It is however desirable for the network to avoid creating nodes with critically low energy reserves.
{\hide{Specifically, nodes with lower residual energy can {only} contribute in relaying and/or generating future messages for shorter durations. In the extreme case, {having} a sizable fraction of  depleted nodes can gravely jeopardize the functionality of the network.}}
We capture the impact of a forwarding policy on the residual energy reserves of the nodes by penalizing the nodes that have lower energy levels. Specifically, denoting the terminal time as $T$, the overall penalty associated with the distribution of the residual energies of nodes at $T$, henceforth referred to as the {\em energy cost} of the policy, is captured by:
\hide{\begin{gather}}
$\sum_{i=0}^B a_i\left(S_i(T)+I_i(T)\right)$,
\hide{\end{gather}}
in which, $\{a_i\}$ is a \emph{decreasing} sequence in $i$, i.e., a higher penalty is associated with lower residual energies at $T$.

{The trade-off can now be stated as follows:} by using a more aggressive forwarding policy (i.e., higher $u_i(t)$'s and for longer durations),
the message propagates faster and there is a greater chance of delivering the message to 
the destination {in a timely manner}.
However, this will lead to lesser overall remaining energy in the nodes upon delivery of the message, and it will potentially  push the energy reserves of some nodes to critically low levels, degrading the future performance of the  network.

\paragraph*{Overall Objective and Problem Statements}
We now  state the two optimization problems {for which we provide necessary structural results for optimal forwarding policies in \S\ref{sec:Structural_Results}.}

\textbf{Problem 1: Fixed Terminal Time}
Considering a fixed terminal time $T,$ we seek to maximize the following utility\hide{ overall utility function}:
  \begin{gather}
  R=- \sum_{i=0}^B a_i\left(S_i(T)+I_i(T)\right)\label{objective_single}
\end{gather}
 by dynamically selecting the vector  $\bu(t) = (u_{\taus}(t),\ldots,u_{B}(t))$ of piece-wise continuous controls subject to control constraints $0\leq u_{i}(t)\leq 1$ for all $\taus \leq i\leq B$, $0\leq t\leq T$ and throughput constraint~\eqref{path_constraint}. States $\bS(t)$ and $\bI(t)$
 satisfy state  dynamics~\eqref{model:no_recharge} and positivity and normalization conditions~\eqref{states:constraints}.  

\textbf{Problem 2: Optimal Stopping Time} We seek to  minimize {a combination of} a penalty
associated with the terminal time $T$ {(the time taken to satisfy the throughput constraint~\eqref{path_constraint})} and one associated with the adverse effects on the residual energy of nodes through choice of the control $\bu$.
We represent the  penalty  associated with terminal time $T$ as $f(T)$. We make the natural
assumption that $f(T)$ is \emph{increasing} in $T$.  We further assume that $f(T)$ is differentiable  (thus $f'(T)>0$){\hide{ and $f(T)$ is finite for all finite values of $T$}}. Considering a variable terminal time $T$ that
is selected to satisfy \eqref{path_constraint} with equality, the system seeks  to maximize\hide{ the overall utility function}:
\begin{equation}\label{objective_single_stopping}
R= \hide{\max_{0\leq u_i\leq 1}} -f(T)-\sum_{i=0}^Ba_i(S_i(T)+I_i(T))
\end{equation}
by dynamically regulating the piecewise continuous set of controls $\bu(t) = (u_{\taus}(t),\ldots,u_{B}(t))$ subject to the control constraints $0\leq u_{i}(t)\leq 1$ for all $\taus \leq i\leq B$, $0 \leq t \leq T$ and \eqref{omit}.  As in \hide{fixed terminal time}Problem 1,  states $\bS(t)$ and $\bI(t)$
 satisfy state  dynamics~\eqref{model:no_recharge} and positivity and normalization conditions~\eqref{states:constraints}.

\section{Optimal Forwarding Policies}\label{sec:Structural_Results}
We identify the structure of the optimal forwarding policies in \S\ref{sec:mainresults} and prove them in \S\ref{sec:proof1} and \$\ref{sec:proof2} respectively. {Our theorems apply to both the Fixed Terminal Time and Optimal Stopping Time problem statements.}

\subsection{Structure of the optimal controls}
\label{sec:mainresults}
We establish that the optimal dynamic forwarding policies require the nodes to opportunistically forward the message to any node that they encounter until a threshold time that depends on their current remaining energy.\footnote{As an infective node transmits, its energy level sinks;  the threshold of each infective {node }should therefore be measured with regards to the \emph{residual} level of energy (and not, for example, the starting level).} Once the threshold is passed, they cease {\hide{to forward}forwarding} the message until the time-to-live (TTL) of the message.
In the language of control theory, we show that, excluding the optimal controls related to energy levels for which the fraction of infectives is zero throughout, all optimal controls are bang-bang with at most one jump from one to zero. In the excluded cases, optimal controls do not affect the evolution of states or objective values. 

\begin{Theorem}\label{Thm:Gen_structure}
Suppose the set $\mathcal{U}^*$ of optimal controls is not empty.\footnote{If $\mathcal{U}^*$  is non-empty, the problem is feasible, i.e., there exists at least one control for which the throughput constraint holds. But, even if the problem is feasible,  $\mathcal{U}^*$  may be empty, albeit rarely. For example,  there may be an infinite  sequence of optimal controls such that the objective values constitute a bounded increasing sequence of positive real numbers; such a sequence will have a limit but the limiting value may not be attained by any control.}  
Then for all optimal controls $\bu$ in $\mathcal{U}^*$, and for all $\taus\leq i\leq B$ such that $I_i \not \equiv 0$, there exists a $t_i \in [0, T]$  such that   $u_i(t)=1$ for $0 < t < t_i$  and $u_i(t) = 0$ for $t_i < t \leq T.$\footnote{Since the optimal controls
associated with energy levels for which the population of the infectives is zero throughout do not influence the evolution of states or the objective values, this theorem implies that unless $\mathcal{U}^*$ is empty, there exists an optimal control in $\mathcal{U}^*$ that will have the reverse-step function  structure posited in the theorem for all $\taus\leq i\leq B$. Note that the irrelevance of optimal controls associated with energy levels with zero population of infectives implies that the optimal controls are not, in general, unique.} Moreover, under any optimal control, for all $\taus\leq i\leq B$, either $I_i(t) = 0$ for all $t \in [0, T]$ or
$I_i(t) > 0$ for all $t \in (0, T].$
\end{Theorem}

\begin{figure}[htb]
 \centering
\includegraphics[width=0.95\columnwidth]{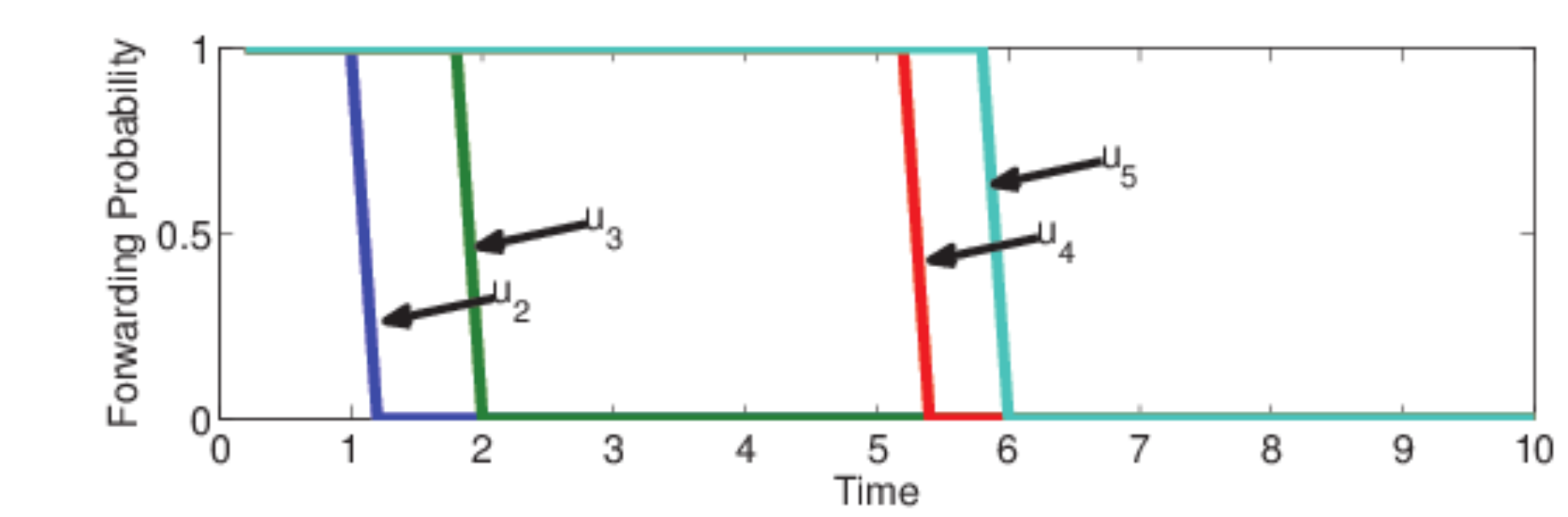}
\subfigure{{\scriptsize (a) convex terminal-time penalties}}
\hide{\subfigure{{\scriptsize  (a) Fixed terminal time. \hide{The mandated delivery probability was 90\%. The initial distribution is $\mathbf{I_0}=(0,0,0,0,0,0.05)$ and $\mathbf{S_0}=(0,0,0,0.55,0.3,0.1)$, and the terminal-time  penalties are
 $a_i= (B-i)^2$.}}}}
\hide{\label{fig:example}
\end{figure}
\begin{figure}[htb]
 \centering
\includegraphics[scale=0.5]{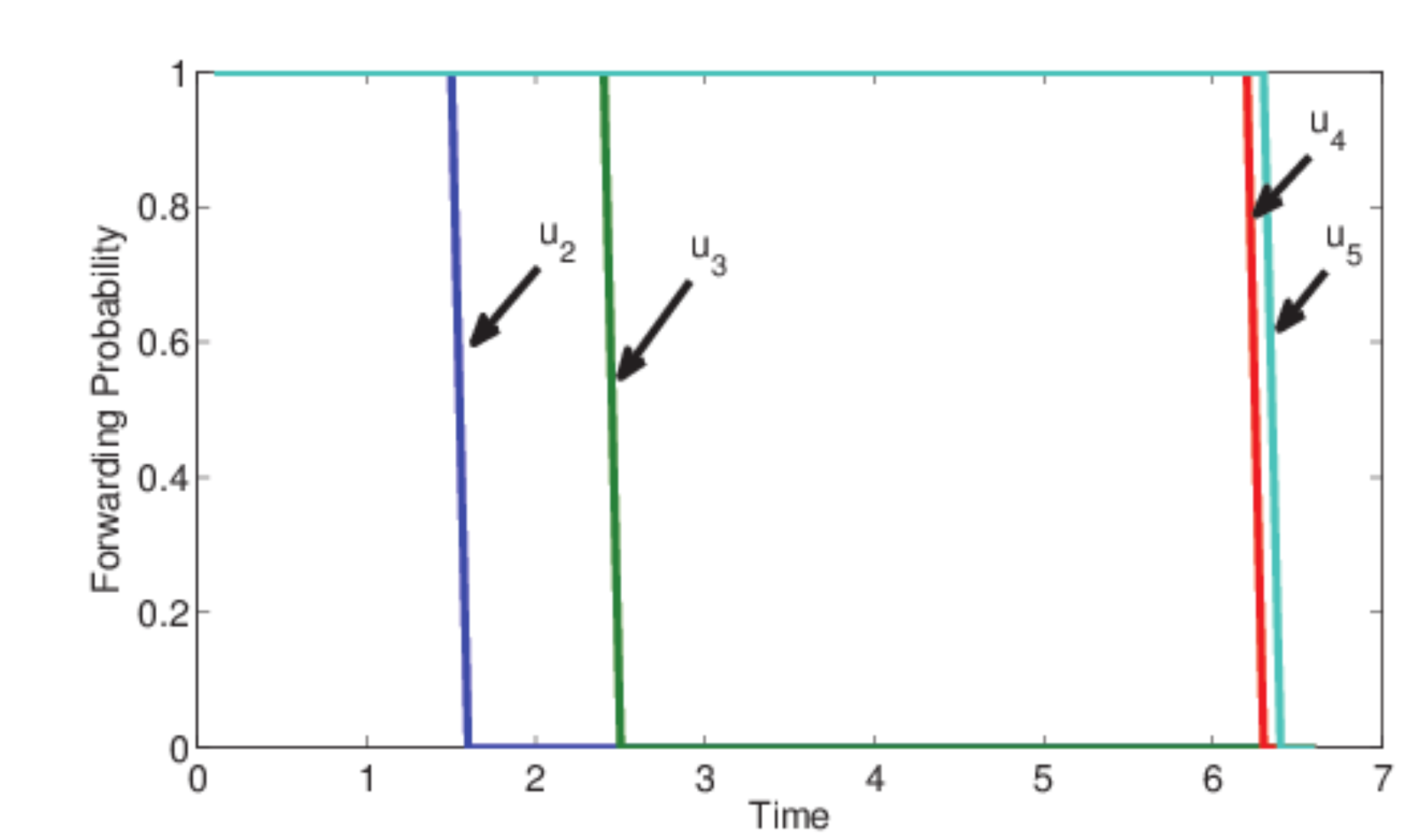}
\subfigure{{\scriptsize  (b) Optimal stopping time. \hide{The mandated delivery probability was 95\%. The initial distribution is $\mathbf{I_0}=(0,0,0,0,0,0.1)$ and $\mathbf{S_0}=(0,0,0,0.5,0.3,0.1)$, and the battery penalties are
 $a_i= (B-i)^2$, with $f(T)=T^2$.}}}}
\includegraphics[width=0.95\columnwidth]{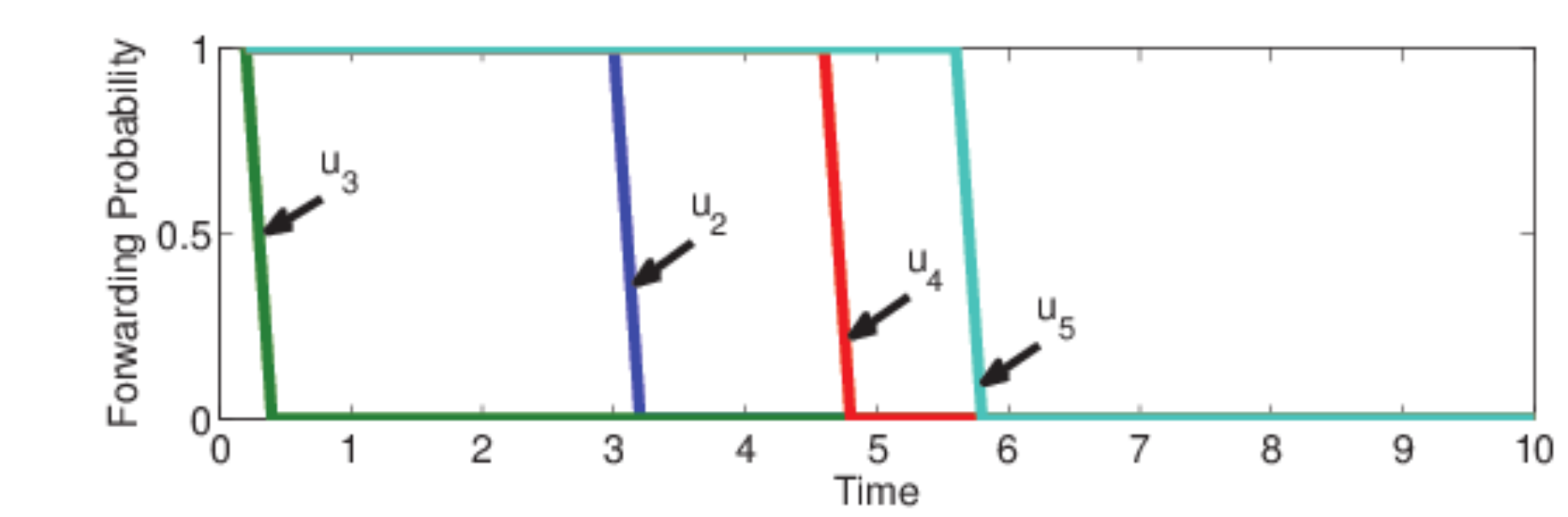}
\subfigure{{\scriptsize (b) non-convex terminal-time penalties}}
\caption{Illustrative examples for Theorems~\ref{Thm:Gen_structure} and~\ref{Thm:Order} for the fixed terminal time problem.\hide{(a) fixed terminal time and (b) optimal stopping time  problems.} The controls are plotted for a system with parameters:
$B=5$, $r=1$, $\taus=2$, $\beta=\beta_0=2$, $T=10$, and $\mathbf{S_0}=(0,0,0,0.55,0.3,0.1)$, with the mandated probability of delivery being 90\%. In (a), the terminal time penalty sequence was $a_i= (B-i)^2$ and $\mathbf{I_0}=(0,0,0,0,0,0.05)$, while in (b) the terminal time penalty sequence was $a_0=4.4$, $a_1=4.2$, $a_2=4$, $a_3=1.2$, $a_4=1.1$, $a_5=1$ (i.e., the $\{a_i\}$ sequence was neither convex nor concave) and $\mathbf{I_0}=(0,0,0,0,0.025,0.025)$.}\hide{. In (a),, and the initial states were $\mathbf{I_0}=(0,0,0,0,0.025,0.025)$ and $\mathbf{S_0}=(0,0,0,0.55,0.3,0.1)$\hide{. In (a)}, $T=10$, $\mathbf{I_0}=(0,0,0,0,0,0.05)$ and $\mathbf{S_0}=(0,0,0,0.55,0.3,0.1)$, with the mandated probability of delivery being 90\%.\hide{ In (b), $f(T)=T^2$,  $\mathbf{I_0}=(0,0,0,0,0,0.1)$ and $\mathbf{S_0}=(0,0,0,0.5,0.3,0.1)$, with the mandated probability of delivery being 95\%.}}
\label{fig:example}
\end{figure}

\hide{We illustrated Theorems~\ref{Thm:Gen_structure} and~\ref{Thm:Order} for the fixed terminal time case in Fig.~\ref{fig:example}. Now, we show the same for the optimal stopping time problem (Fig.~\ref{fig:examplestop}). As we can see, Theorems~\ref{Thm:Gen_structure} and~\ref{Thm:Order} also hold for this example.

Furthermore, Fig.~\ref{fig:nonconvexstop} shows that the convexity of the final state penalties is also a necessary assumption for Theorem~\ref{Thm:Order} to hold, as it was shown in Fig.~\ref{fig:nonconvex} for the fixed terminal time case. Here, we see that a non-convex sequence of penalties leads to a non-monotonic sequence $\{t_i\}$.}

Given any optimal control $\bu$, we define a set $\mathcal{Z}(\bu)$ such that $\mathcal{Z}(\bu) = \{i: \taus \leq i \leq B, I_i(t) > 0, \forall t \in (0, T]\}.$
{The above theorem implies that the population of the infectives is zero throughout
for any index  outside $\mathcal{Z}(\bu)$ (i.e.,
 if $i \not\in \mathcal{Z}(\bu), I_i(t) = 0$ for all $t \in [0, T]$), and we therefore characterise the optimal control only for the indices that are in $\mathcal{Z}(\bu).$}
Also, for each $i \in \mathcal{Z}(\bu)$, $t_i$ is the threshold time associated with the optimal control $u_i$. Intuitively, we would expect each optimal control to be a non-increasing function of time, since  if a control is increasing over an interval, flipping that part of the control in time would result in
 earlier propagation of the message and a higher throughput with the same final state energies. The theorem, however, goes beyond this
intuition in  that it establishes that
  optimal controls    are at their maximum value up to certain threshold times  and then drop abruptly to zero (Fig.~\ref{fig:example}-(a)).  For the fixed terminal time problem, the optimal controls can therefore be  represented as a vector of $B-\taus+1$  threshold times corresponding to different energy levels. This vector can be calculated through an optimization in  the search space of $[0,T]^{B-\taus+1}.$  For the optimal stopping time problem, there is an additional degree of freedom; the stopping time $T$ itself. Note that   $T\in [0,T_0]$, where $T_0$ satisfies \eqref{path_constraint} with equality if all controls are always zero. This is because no optimal control can have $T>T_0$, as in that case both the energy cost and the penalty associated with terminal time will exceed  that of the all-zero controls case. Thus,  the optimal stopping time and the thresholds can be calculated through an optimization in  the space of $\{(T, \underline{t}): 0 \leq T \leq T_0,~ \underline{t} \in[0,T]^{B-\taus+1}\}$.{\hide{The thresholds (and the stopping time for the optimal stopping time problem)  may be computed once at the source node of the message and added to the message as a small overhead. Each node that receives the message simply retrieves the threshold times and forwards the message if its age is less than the threshold entry corresponding to its residual energy level. The execution of the policy at each node is therefore simple  and based only on local information.}} The one-time calculation of the threshold levels (and the optimal stopping time as appropriate) at the origin can be done by estimating the \hide{current distribution of energy in the network which only relies on knowledge of}
 the fractions of nodes with each energy level irrespective of their identities. This data can then be added to the message as a small overhead. 
{Therefore, optimal message forwarding has the following structure:

\begin{algorithm}
\caption{Source Node}
\label{alg1}
\begin{algorithmic}[1]
\State {\bf Given:} $\mathbf{I}_0:=(I_{00},\ldots,I_{0B})$.
\State Estimate the distribution of energy among nodes.
\State Find the best set of thresholds $\{t_i\}$ (and optimal stopping time $T$ in the optimal stopping time problem).
\State Append the header, which contains the destination, $T$, and $\{t_i\}$, to the message. 
\State Create an initial distribution of messages such that for $j=0,\ldots,B$, infectives of energy level $j$ constitute a fraction $I_{0j}$ of the whole population.
\end{algorithmic}
\end{algorithm}
\begin{algorithm}
\caption{Infective Nodes}
\label{alg2}
\begin{algorithmic}[1]
\State On receipt of the message, extract destination, thresholds $\{t_i\}$, and stopping time $T$ from the header.
\State Measure own residual energy $i$.
\While {$i \geq \taus$ {\bf and} $t\leq T$}
\If{node $n$ encountered} \State query its state [low cost].
\If{{$n$ = \{destination\}}} 
\If{$n$ has not received the message yet}
\State{transmit the message.}
\EndIf
\State exit.
\ElsIf{$n$ = \{$S$ with energy $j\geq r$\}  {\bf and} $t<t_i$} 
\State forward message.
 \State $i \leftarrow (i-\taus)$.
\EndIf
\EndIf
\EndWhile  
\end{algorithmic}
\end{algorithm}
}
  Intuitively, it appears that the threshold-times will
be non-decreasing functions of the energy levels, since lower levels of residual
energy are penalized more and
 the energy consumed in each transmission
and reception is the same irrespective of  the energy levels of
the nodes. The optimal controls depicted in  Fig.~\ref{fig:example}-(a) suggest the same:
$ t_2 < t_3 < t_4 < t_5.$  We now confirm the above intuition in the case that the terminal-time penalty
sequence $\{a_i\}$ satisfies certain properties:

\begin{Theorem}\label{Thm:Order}
 Assume that the sequence $\{a_i\}$ in~\eqref{objective_single} is strictly convex.\footnote{A sequence $\{a_i\}$ is strictly convex if the difference between the
penalties associated with consecutive energy levels
increases with a decrease in energy levels (mathematically, for each {\hide{$\taus + 2 \leq i \leq B$}$2 \leq i \leq B$},  {$a_{i-1} - a_i < a_{i-2} - a_{i-1}$}). A sequence $\{a_i\}$ is strictly concave if the difference between the
penalties associated with consecutive energy levels
decreases with an increase in energy levels (mathematically, for each {\hide{$\taus + 2 \leq i \leq B$}$2 \leq i \leq B$},  {$a_{i-1} - a_i > a_{i-2} - a_{i-1}$}). } Then, for any optimal control $\bu$, the sequence $\{t_i\}$ for $i \in \mathcal{Z}(\bu)$ is non-decreasing in $i$.
\end{Theorem}

Fig.~\ref{fig:example}-(a) {illustrates} the threshold times for a strictly convex and decreasing
sequence of terminal penalties. The naive intuition provided before Theorem~\ref{Thm:Order} will however mislead us in general {---} we now present {examples that show when the strict convexity requirement of the terminal-time penalty sequence is not satisfied, the claim of the theorem may not hold.} One sample configuration is when we have a sharp reduction in penalty  between two consecutive final energy levels, with penalties on either side being close to each other, e.g., $a_0 \approx a_1\approx a_2 \gg a_3 \approx a_4 \approx a_5$ in Fig. \ref{fig:example}-(b).
The motivation for such a setting could be the case where the system is primarily interested in ensuring that it  retains  a certain minimum amount of energy (e.g., $3$ units in Fig.~\ref{fig:example}-(b)) at the terminal time: energy values above the requisite threshold (e.g., $4, 5$ in Fig.~\ref{fig:example}-(b)) acquire insignificant additional rewards and energy values
below the threshold (e.g., $0, 1, 2$ in Fig.~\ref{fig:example}-(b)) incur insignificant additional penalties, but the penalty at the threshold amount is substantially lower than that at the next lowest value.  Fig.~\ref{fig:example}-(b) reveals that Theorem~\ref{Thm:Order} need not  hold for such a setting, as nodes with energy values that are either higher or lower than $3$ would be incentivized to propagate the message (because of the low loss incurred for propagation in terms of final states), but those with exactly $3$ units of energy would
be extremely conservative, as there is a large penalty associated with any further propagation of the message. Thus, $t_3 < \min(t_2, t_4, t_5).$ The sequence of terminal-time penalties in  Fig.~\ref{fig:example}-(b) is neither convex nor concave.  But, Theorem~\ref{Thm:Order} does not hold for
 concave terminal-time penalties either (Table~\ref{table:anomaly}). {Therefore, the convexity of the terminal-time penalty sequence is integral to the result of Theorem~\ref{Thm:Order}. }

\begin{table}[bth]
\begin{center}
\begin{tabular}{ |c | c | c | }
   \hline
   & \multicolumn{2}{|c|}{Threshold Times of Controls }\\
   \hline
   & Energy Level 4& Energy Level 5\\ \hline \hline
   $\alpha=0.5$ & 5.75 & 1.75 \\
   $\alpha=1.5$ & 2.5 & 2.75 \\
   $\alpha=2$ & 2.5 & 2.75 \\
   \hline
\end{tabular}
\end{center}
\begin{center}
\caption{{\scriptsize An example for non-ordered {threshold times} of the
optimal controls for concave terminal time penalties in the settings of
Theorem~\ref{Thm:Order} for the fixed terminal time problem. The parameters were exactly the same as those used in Fig.~\ref{fig:example}-(a), with the difference that $a_i=(B-i)^\alpha$,   $\alpha$ is varied over the values $\{0.5,1.5,2\}$, $\mathbf{I_0}=(0,0,0,0,0,0.1)$ and $\mathbf{S_0}=(0,0,0,0.3,0.3,0.3)$\hide{ $I_{03}=0.2$, $S_{03}=0.3$, $I_{02}=0.1$, $S_{02}=0.2$,
$I_{01}=0.1$, $S_{01}=I_{00}=S_{00}=0$}.
For $\alpha=0.5$,  the terminal time penalties
become concave, and $t_4 > t_5.$   For $\alpha=\{1.5,2\}$,
the terminal time penalties  are strictly convex, and $t_4 < t_5$ as Theorem~\ref{Thm:Order} predicts.
}}\label{table:anomaly}
\end{center}
\end{table}

\vspace{-0.6in}
 \subsection{Proof of Theorem~\ref{Thm:Gen_structure}}
 \label{sec:proof1}
  We prove Theorem~\ref{Thm:Gen_structure} using tools from classical optimal control theory, specifically Pontryagin's Maximum Principle, which is stated in \S\ref{pontryaginprinciple}. We provide the full proof for the fixed terminal time problem~\eqref{objective_single} in \S\ref{subsec:Fixed_Terminal_time}, and specify the modifications for the optimal stopping time problem in \S\ref{subsubsec:changes_single_stopping}.

{\subsubsection{Pontryagin's Maximum Principle with Terminal Constraint}\label{pontryaginprinciple}

We start by stating the problem for a fixed terminal time $t_1$. Let $\bu^*$ be a piecewise continuous control solving\hide{the problem}:
\begin{align}\label{problem}
\text{maximize} &\int_{t_0}^{t_1} f_0({\bf x}(t), \bu(t),t) + S_1({\bf x}(t_1))\\
&\dot{\bf x}(t) = f({\bf x}(t), \bu(t), t),\quad
\bx(t_0) = \bx^0,\quad\bu \in \mathcal{U}, \notag\\ 
&x_i^1(t_1)=x_i^1 \qquad1\leq i \leq l,\notag\\
&x_i^1(t_1)\geq x_i^1 \qquad l+1\leq i \leq m,\notag\\
&x_i^1(t_1)\text{~free} ~\qquad i=m+1\leq i \leq n,\notag
\end{align}
and let ${\bf x}^*(t)$ be the associated optimal path. Define
 \begin{align}\label{hamiltonian}\ham({\bf x}(t), &\bu(t), {\bf p}(t), t):=\notag\\&p_0  f_0({\bf x}(t), \bu(t),t) + {\bf p}^T(t) f({\bf x}(t), \bu(t), t)
\end{align} to be the Hamiltonian, with ${\bf p}=\{p_i\}_{i=1}^n$.

\begin{Theorem}\label{pontryagin}\cite[p.182]{seierstad1987optimal}\hide{\cite[p.111]{grass2008}(\!\cite[p. 217]{stengel1994}, p. 217)}
\hide{Let defined on $[t_0,t_1]$, If $t_1$ is fixed, then}There exist a constant $p_0$ and a continuous and piecewise continuously differentiable vector function ${\bf p}(t)=(p_1(t),\ldots, p_n(t))$ such that for all $t \in [t_0,t_1]$, 
\begin{align}
 (p_0,p_1(t),\ldots,p_n(t))\neq \vec{0}\label{vec_neq_zero},
\\\ham({\bf x}^*, \bu^*, {\bf p}(t), t) \geq \ham({\bf x}^*, \bu, {\bf p}(t), t)\qquad\forall{\bu \in \mathcal{U}}\label{maximization}.
\end{align}

Except at the points of discontinuities of $\bu^*(t)$, for $i=1,\ldots,n$: $\dot{p}_i(t) = - \dfrac{\partial \ham({\bf x}^*, \bu^*, {\bf p}(t), t)}{\partial x_i}$.

Furthermore, $p_0=0$ or $p_0=1$, and, finally, the following transversality conditions are satisfied,
\begin{flalign}\label{transversality}
&p_i(t_1) \text{ no condition} \qquad\qquad\qquad\quad\qquad\quad  1\leq i \leq  l,\notag\\
p_i(t_1) &-p_0 \dfrac{\partial S_1(\bx^*(t_1))}{\partial x_i}\geq 0\notag\\&\qquad\qquad\quad\qquad(= 0 \text{ if $x^*_i(t_1)>x^1_i$}) \quad  l +1\leq i\leq m,\notag\\
p_i(t_1) &- p_0 \dfrac{\partial S_1(\bx^*(t_1))}{\partial x_i}=0 \qquad~\quad\qquad\quad~  m+1\leq i \leq  n.
\end{flalign}

\end{Theorem}
Now, we state the analogous theorem when $t_1$ is not fixed in ~\eqref{problem}, and $S_1(\bx(t_1))$ is replaced with $S_1(\bx(t_1), t_1)$, allowing explicit dependence of the cost on the terminal time:
{
\begin{Theorem}\label{pontryagin1}~\cite[p.183]{seierstad1987optimal} 
 Let $(\bx^*(t), \bu^*(t), t_1^*)$ be an admissible triple solving~\eqref{problem} (with $S_1(\bx(t_1), t_1)$) with $t_1\in [T_1, T_2]$, $t_0\leq T_1<T_2$, $T_1, T_2$ fixed. Then the conclusions in Theorem~\ref{pontryagin} hold, with $S_1(\bx^*(t_1),t_1)$ replacing $S_1(\bx^*(t_1))$, and with the addition that 

\begin{align}\label{H(T)equals}
\ham({\bf x}^*, \bu^*, {\bf p}, t_1^*) + p_0 \dfrac{\partial S_1(\bx^*(t_1),t_1)}{\partial t}\begin{cases}
\leq 0 &\text{if~} t_1^*=T_1\\
= 0 &\text{if~} t_1^*\in(T_1, T_2)\\
\geq 0 &\text{if~} t_1^*=T_2
\end{cases}.
\end{align}

%
%
%
%
\end{Theorem}
}
}

\subsubsection{Fixed Terminal Time Problem}\label{subsec:Fixed_Terminal_time}

{For every control $\tilde{\bu}$, we define $\tau_i(\bI(0),\bS(0),\tilde{\bu})\in[0,T]$ as follows: If $I_i(0)>0$, and therefore $I_i(t)>0$ for all $t>0$ due to Theorem~\ref{thm:constraints}, we define $\tau_i(\bI(0),\bS(0),\tilde{\bu})$ to be 0. Else, $\tau_i(\bI(0),\bS(0),\tilde{\bu})$ is the maximum $t$ for which $I_i(t)=0$.}  It follows from Theorem~\ref{thm:constraints} that $I_i(t) = 0$ for all $t \leq \tau_i(\bI(0),\bS(0),\tilde{\bu})$ {and all $i$ such that $I_i(0) = 0$}, and $I_i(t) > 0$ for all $\tau_i(\bI(0),\bS(0),\tilde{\bu}) < t \leq T$. We begin with the hypothesis that there exists at least one optimal control, say $\tilde{\bu}\in \mathcal{U}^*$, and  construct  a control $\bu$ that chooses $u_i(t):= 0$
for $t \leq \tau_i(\bI(0),\bS(0),\tilde{\bu})$ and $u_i(t):= \tilde{u}_i(t)$ for $t > \tau_i(\bI(0),\bS(0),\tilde{\bu}).$ Clearly, the states $\bS(t), \bI(t)$ corresponding to
$\tilde{\bu}$ also constitute the state functions for $\bu${, as the state equations only differ at $t=0$, a set of measure zero}. Thus, $\bu$ is also an optimal control, and
$\tau_i(\bI(0),\bS(0),\tilde{\bu}) = \tau_i(\bI(0),\bS(0),\bu)$ for each $i.$ Henceforth, for notational convenience, we will refer to $\tau_i(\bI(0),\bS(0),\tilde{\bu}), \tau_i(\bI(0),\bS(0),\bu)$ as $\tau_i.$ Note that the definition
of this control completely specifies the values of each $u_i$ in  $[0, \tau_i].$ We will prove
the following lemmas.
\begin{Lemma}\label{lem:reversestepfunction}
 For each $\taus\leq i\leq B$, if  $\tau_i < T$  there exists a $t_i \in [\tau_i, T]$ such that $u_i(t) = 1$ for $\tau_i < t < t_i$ and $u_i(t) = 0$ for $t > t_i.$
 \end{Lemma}
\begin{Lemma}\label{lem:tau}
For all $\taus\leq i\leq B,$  $\tau_i \in \{0,T\}$. 
\end{Lemma}

  If $\tau_i = 0$ for some $i\geq \taus$,  $\tilde{u}_i(t) = u_i(t)$, and $I_i(t) > 0$,  for all $t \in (0, T]$. If $\tau_i=T$, $I_i(t) = 0$ for all $t \in [0, T].$ So the theorem follows from these lemmas, which we prove next.

\paragraph{Proof of   Lemma~\ref{lem:reversestepfunction}}
The lemma clearly holds if $\bu \equiv 0$ (with $t_i = \tau_i$ for all $i \geq \taus$); we therefore consider the case that
$\bu \not\equiv 0$.\footnote{Note that $\bu \equiv 0$ in $(\tau_i, T]$ does not imply $\tau_i=T$.} We proceed in the following steps:

{\bf 1})  Applying standard results from optimal control theory,  we show that
 each optimal control $u_i$ assumes the maximum value ($1$) when a   \emph{switching function} (denoted $\varphi_i$) is positive and the minimum value ($0$)  when
the switching function is negative. However, standard optimal control results do not specify
the nature of the optimal control when the corresponding switching function is at $0$ or the durations for which the switching function is positive, zero, or negative.
 The next step answers these questions using specifics of the problem.

{\bf 2})\hide{\begin{Lemma}\label{lem:varphi_zero_X}
Let $\bu \not\equiv 0.$
For all $i \geq  \taus$, if $\varphi_i (t')=0$ for $t'\in(\tau_i,T)$, then $\varphi_i (t)<0$ for all $t>t'$. Also, if $\varphi_i(T)=0$,  $\varphi_i(t)>0$ for $t \in (\tau_i, T).$
\end{Lemma}}
The switching functions turn out to be  continuous functions of time.
{We want to show that for each $i \geq \taus$, there exists $t_i \in [\tau_i, T]$ such that
the relevant switching function ($\varphi_i$) is positive for $t \in (\tau_i, t_i)$,  negative for $t \in (t_i, T]$, and equal to zero at $t_i$ only if $t_i\in (\tau_i,T)$.  } Lemma~\ref{lem:reversestepfunction} now follows from the relation between the optimal control and the switching function obtained in the first step.\footnote{Note that we still do not know the value of $u_i$ at  the time epoch  $t_i$ at which the corresponding switching function $\varphi_i$ {\hide{is}may be} zero.  This is not{, however,} a serious deficiency since the value of the optimal control in any set of measure zero does not affect the state evolution.}

\noindent{\bf Step 1} Consider the system in~\eqref{model:no_recharge} and the
objective function in~\eqref{objective_single}.
To make the analysis more tractable, we introduce the following new state variable:
{$
 \dot{E}:=\sum_{i=\taus}^B I_i$, with $E(0):=0
$.}

Therefore, our {throughput} constraint~\eqref{path_constraint} simply becomes:
$E(T)\geq  -\ln (1-p)/\beta_0.$ 

{ To facilitate an appeal to Theorem~\ref{pontryagin}, we take $\bx^T = (E,\bS^T,\bI^T)$, $\bu = \bu$, $p_0=\bar{\lambda}_0$, ${\bf p}= (\lambda_E, {\boldsymbol{\lambda}}, {\boldsymbol\rho})$, $l=0$, $m=1$, $x_1^1=  -\ln (1-p)/\beta_0$, $f_0 \equiv 0$, $t_0=0$, $t_1 = T$, and $S_1(\bx^*(t_1)) = R$, the optimization objective. In this case, $\{f_i\}_{i=1}^{2N+3}$ are given by the $\dot{E}$ equation above and by \eqref{model:no_recharge}.}

Using these replacements, the Hamiltonian~\eqref{hamiltonian} 
 becomes
\begin{equation}\label{define:Hamiltonian:Gen}
 \begin{split}
\ham= -\sum_{i=r}^B[\beta \lambda_iS_i \sum_{j=\taus}^{B}u_jI_j]
  +\sum_{i=r}^{B}[\beta \rho_{i-r}S_{i} \sum_{j=\taus}^{B}u_jI_j]\\
  +\sum_{i=\taus}^{B}[\beta u_{i}\rho_{i-\taus} I_{i} \sum_{j=r}^{B}S_j]
  -\sum_{i=\taus}^B[\beta u_i\rho_iI_{i} \sum_{j=r}^{B}S_j ]+ \lambda_E\sum_{i=\taus}^B I_i
\end{split}
\end{equation}
where, at the points of continuity of the controls, the absolutely continuous co-state functions $\lambda_i$, $\rho_i$ and $\lambda_E$ satisfy
\begin{alignat}{2}
 &\dot{\lambda}_i =-\frac{\partial\ham}{\partial S_i}=\beta\lambda_i\sum_{j=\taus}^Bu_jI_j-\beta\rho_{i-r}\sum_{j=\taus}^Bu_jI_j \quad~~& \notag\\
&\qquad-\beta\sum_{j=\taus}^{B}u_{j}\rho_{j-\taus}I_{j}+\beta\sum_{j=\taus}^Bu_j\rho_jI_j & (r\leq i\leq B)\notag\\
 &\dot{\lambda}_i=-\frac{\partial\ham}{\partial S_i}=0& (i<r)&\notag\\
&\dot{\rho}_i=-\frac{\partial\ham}{\partial I_i}=\beta u_i\sum_{j=r}^B\lambda_jS_j+\beta u_i\rho_i\sum_{j=r}^BS_j&\notag\\
&\qquad-\lambda_E-\beta u_i\sum_{j=r}^{B}\rho_{j-r}S_{j}-\beta u_{i}\rho_{i-\taus}\sum_{j=r}^BS_j& (\taus\leq i\leq B)\notag\\
&\dot{\rho}_i=-\frac{\partial\ham}{\partial I_i}=0 & (i<\taus)\notag\\
&\dot{\lambda}_E= -\frac{\partial \ham}{\partial E}=0 & \label{eq:co-states:Gen}
\end{alignat}
with the final constraints:
\begin{equation}\label{co_st_finals:Gen}
 \begin{split}
 &\lambda_{i}(T)=-\bar{\lambda}_0 a_i, \quad \rho_i(T)=- \bar{\lambda}_0 a_i, \quad \forall i=0,\ldots,B\\
&\lambda_E(T)\geq 0,\quad {\lambda_E(T) \left[E(T)+ \ln(1-p)/\beta_0\right]=0},
\end{split}
\end{equation}
and $\bar{\lambda}_0\geq 0$.

\hide{We now appeal to \emph{Pontryagin's maximum principle with terminal constraints:}}

We formally define the switching functions $\varphi_i$ as follows:
\begin{align} \label{eq:varphi_simpler:Gen}
\varphi_i:=\frac{\partial \mathcal{H}}{\partial u_i}=
\beta I_i\Bigg[\sum_{j=r}^B\left(-\lambda_j+\rho_{j-r}+\rho_{i-\taus}-\rho_i\right)S_j\Bigg],\notag\\\qquad\qquad\qquad\qquad\qquad(\taus\leq i\leq B).
\end{align}

Note that $\varphi_i$ is a continuous function of time for each $\taus\leq i\leq B.$ Also, we have:
{\begin{align}\label{hamsimpler}
\ham=\lambda_E\sum_{i=\taus}^B I_i + \sum_{i=\taus}^B \varphi_i u_i.
\end{align}}
From Theorem~\ref{pontryagin},  maximizing the Hamiltonian~\eqref{maximization} yields
\begin{align}\label{optimal_u_i:Gen}
 u_i(t)=
\begin{cases}1 \ \ \text{for}\ \ \varphi_i(t)>0\qquad \\
 0 \ \ \text{for}\  \    \varphi_i(t)<0.
\end{cases}
\end{align}
Furthermore, $\varphi_i(t)u_i(t) \geq 0$ for each $\taus\leq i\leq B$ and all $t \in [0, T]$;  otherwise
the value of the Hamiltonian can be increased at $t$ by choosing $u_i(t) = 0$.

Equations (\ref{eq:varphi_simpler:Gen}, \ref{optimal_u_i:Gen}) reveal an accessible intuition about the logic behind the decision process: at any given time, by choosing a non-zero $u_i$,  infectives with energy level $i\geq \taus$ forward the message to susceptibles of {any energy level $j\geq r$ and turn into infectives with $i-\taus$ energy units, with the susceptibles turning into infectives of energy level $j-r$.}  The optimal control determines whether such an action is \emph{beneficial}, taking into account the advantages (positive terms) and disadvantages (negative terms).


\noindent{\bf Step 2} {To establish this claim, we prove the following lemma:} \begin{Lemma}\label{lem:varphi_zero_X}
Let $\bu \not\equiv 0.$
For all $i \geq  \taus$, if $\varphi_i (t')=0$ for $t'\in(\tau_i,T)$, then $\varphi_i (t)<0$ for all $t>t'$. Also, if $\varphi_i(T)=0$,  $\varphi_i(t)>0$ for $t \in (\tau_i, T).$
\end{Lemma}
 For any $i\geq\taus$, we show that for any $t \in (\tau_i,T)$ at which $\varphi_i(t)=0$,
$\dot{\varphi}_i(t^+) < 0$ and $\dot{\varphi}_i(t^-) < 0$.\footnote{$x(a^+)= \lim_{t\downarrow a}x(t)$, $x(a^-)= \lim_{t\uparrow a}x(t)$.} Furthermore, we show that if $\varphi_i(T)=0$,
$\dot{\varphi}_i(T^-) < 0$. We state and prove a property of real-valued functions which we will use in proving Lemma~\ref{lem:varphi_zero_X} from the above.

\begin{Property}\label{property2}
If $g(x)$ is a continuous and piecewise differentiable function over $[a,b]$ such that $g(a)=g(b)$ while $g(x)\neq g(a)$ for all $x$ in $(a,b)$, $\frac{dg}{dx}(a^+)$ and $\frac{dg}{dx}(b^-)$ cannot be negative simultaneously.
\end{Property}

\begin{proof}
We denote the value of $g(a)$ and $g(b)$ by $L$. If  $\frac{dg}{dx}(a^+)<0$, there exists $\epsilon>0$ such that $g(x)<L$ for all $x \in (a,a+\epsilon)$, and if $\frac{dg}{dx}(b^-)<0$, there exists $\alpha>0$ such that $g(x)>L$ for all $x \in (b-\alpha,b)$. Now $g(a~+~\frac{\epsilon}{2})<L$ and $g(b-\frac{\alpha}{2})>L$; thus, due to the continuity of $g(t)$, the intermediate value theorem states that there must exist a $y \in (a+\frac{\epsilon}{2},b-\frac{\alpha}{2})$ such that $g(y)=L$. This contradicts $g(x)\neq g(a)$ for $x \in (a,b)$. The property follows.
\end{proof}

If $\varphi_i(t) = 0$ and $\dot{\varphi}_i(t^+) < 0$ for $t<T$, we have $\varphi_i(t+\Delta{t})= \varphi_i(t)+\int_{t}^{t+\Delta{t}}  \dot{\varphi}_i(x) \, \mathrm{d} x=\int_{t}^{t+\Delta{t}}  \dot{\varphi}_i(x) \, \mathrm{d} x$, which proves the existence of an interval $(t,t+\epsilon]$ over which $\varphi_i$ is negative. If $t+\epsilon\geq T$, then the claim holds, otherwise there must exist a $t'$, $t<t'\leq T$ such that $\varphi_i({t'})=0$ and $\varphi(\bar{t})\neq0$ for $t<\bar{t}<t'$ (from the continuity of $\varphi_i(t)$). Note that because $\varphi_i({t'})=0$, we have $\dot{\varphi_i}({t'^-})<0$. This contradicts Property~\ref{property2}, thereby completing the proof of the first part of the lemma.
For the second part, note that if $\varphi_i(T)=0$ and $\dot{\varphi}_i(t^-)<0$, there exists an interval $(T-\epsilon, T)$ over which $\varphi_i$ is {\hide{negative}positive}. If $T-\epsilon\leq \tau_i$, then the claim holds, otherwise there must exist a {$t' \in (\tau_i, T)$} such that $\varphi_i(t')=0$ and $\varphi(\bar{t})\neq0$ for $t'<\bar{t}< T$ (from the continuity of $\varphi_i(t)$). Note that because $\varphi_i({t'})=0$, as we show we have $\dot{\varphi_i}({t'^+})<0$. This contradicts Property~\ref{property2}, thereby completing the proof of the second part of the lemma.


We now seek to upper bound $\dot{\varphi}_i(t^+)$ and  $\dot{\varphi}_i(t^-)$ for $t \in ( \tau_i, T)$  at which $\varphi_i(t)=0$,
and subsequently prove that the upper bound is negative. For $t=T$, we only consider the left hand limit of the derivative. Keeping in mind that $I_i(t) > 0$ for $t > \tau_i$,  
 at any $t > \tau_i$ at which $\bu$ is continuous, \small
\begin{align*}
\dot{\varphi}_i=& \dot{I}_i \frac{\varphi_i}{I_i}-\varphi_i\beta\sum_{j=\taus}^Bu_jI_j
+\beta I_i \sum_{j=r}^B (-\dot{\lambda}_j+\dot{\rho}_{j-r}+\dot{\rho}_{i-\taus}-\dot{\rho}_i)S_j.
\end{align*}
\normalsize

From the expressions for the time derivative of the co-states in~\eqref{eq:co-states:Gen} combined with the expression for the switching functions in~\eqref{eq:varphi_simpler:Gen}, and using (from~\eqref{define:Hamiltonian:Gen}) that
$ \sum_{j=r}^B - \dot{\lambda}_j(t)S_j(t)=\ham(t)-\lambda_E(t)\sum_{j=\taus}^BI_{j}(t),$ we can write:
\begin{align*}
 \dot{\varphi}_i&=\beta I_i  \left(\ham(t)-\lambda_E\sum_{j=\taus}^BI_j - \lambda_E\sum_{j=r}^BS_j\right)\\
 &\hspace{-0.05in}+ \dot{I}_i \frac{\varphi_i}{I_i}-\varphi_i\beta\sum_{j=\taus}^Bu_jI_j+\varphi_{i}u_i\beta\sum_{j=r}^B S_j \\&\hspace{-0.05in}-\beta^2I_i\sum_{j=r}^B S_j u_{j-r} (\sum_{k=r}^B\left[-\lambda_k+\rho_{k-r}+\rho_{j-r-\taus}-\rho_{j-r}\right]S_k)\\
&\hspace{-0.05in}-\beta^2I_i (\sum_{j=r}^B S_j) u_{i-\taus}  (\sum_{k=r}^B\left[-\lambda_k+\rho_{k-r}+\rho_{i-2\taus}-\rho_{i-\taus}\right]S_k).
\end{align*}

Now, consider a $t \in (\tau_i, T)$  at which {$\varphi_i(t)=0$.} We show that  the right and left-hand limits of all terms in the second line are zero at $t$:

{ From the continuity of  $I_i$ {and since $t>\tau_i$}, $I_i(t)>0$. Thus $I_i(t')$ is positive and bounded away from $0$
{for $t'$ in a neighborhood of $t.$} Furthermore, Lemma~\ref{dotI} shows that $|\dot{I}_i(t^+)|$ and $|\dot{I}_i(t^-)|$ exist and are bounded for all $t\in(0,T)$. Thus, from the continuity of $\varphi_i$ \hide{and the state and co-state functions }at $t$ and since $\varphi_i(t)=0$,
\hide{and since the controls are bounded, }$\dot{I}_i(t^+) \frac{\varphi_i(t^+)}{I_i(t^+)}$ and $\dot{I}_i(t^-) \frac{\varphi_i(t^-)}{I_i(t^-)}$
equal zero. Due to Theorem~\ref{thm:constraints}, since the states and controls are bounded and since $\varphi_i(t)=0$, the right hand and left hand limits at $t$ of the second and third terms in the second line are also zero.  We now argue that the right hand and left hand limits of lines 3 and 4 are non-positive.} Starting with line $3$, this is because for $j\geq r$,
\small
  \[ I_{j-r} \left( u_{j-r} \sum_{k=r}^B\left[-\lambda_k+\rho_{k-r}+\rho_{j-r-\taus}-\rho_{j-r}\right]S_k \right) = \varphi_{j-r}u_{j-r}. \]
\normalsize
The right hand side is non-negative at each $t$, as argued after~\eqref{optimal_u_i:Gen}. For $t>\tau_{j-r}$,  $I_{j-r}(t) > 0$. Thus {for all such $t$, \hide{at each $t>\tau_{j-r}$}}
  \[\left( u_{j-r} \sum_{k=r}^B\left[-\lambda_k+\rho_{k-r}+\rho_{j-r-\taus}-\rho_{j-r}\right]S_k \right) \geq 0.\]
   For {$0<t \leq \tau_{j-r}$}, $u_{j-r}(t) = 0.$ Thus, at all {$t>0,$} the above inequality holds.\hide{\[\left( u_{j-r} \sum_{k=r}^B\left[-\lambda_k+\rho_{k-r}+\rho_{j-r-\taus}-\rho_{j-r}\right]S_k \right) \geq 0.\]}
 
 {Now, since $\bI$, $\bS$
are continuous and $\bu$ has right and left hand limits at each t,  the right and left hand limits of the LHS above exist; such limits are clearly non-negative at each t.} The same arguments apply for line $4$ as well (except that $i-\taus$ must be considered instead of $j-r$, with $i \geq \taus$).  It follows that  at any $t > \tau_i$  at which $\varphi_i(t)=0$,
  \begin{align*}
  \dot{\varphi}_i(t^+)\leq&\beta I_i(t^+)  \left(\ham(t^+)-\lambda_E[\sum_{j=\taus}^BI_j(t^+) + \sum_{j=r}^BS_j(t^+)]\right),
\end{align*}
\begin{align*}
 \dot{\varphi}_i(t^-)\leq&\beta I_i(t^-)  \left(\ham(t^-)-\lambda_E[\sum_{j=\taus}^BI_j(t^-) + \sum_{j=r}^BS_j(t^-)]\right).
 \end{align*}
 Using the same arguments it may also be shown that the latter inequality  holds at $t=T$ if $\varphi_i(T) = 0.$

The lemma now follows once we prove (in Appendix-\ref{appendix_lemma_5}):
\begin{Lemma}\label{lem:5} If $\bu \not\equiv 0$,  then
 for all $t\in(0,T)$, we have:
\begin{align}
 \ham(t^-)-\lambda_E(t^-)\left[\sum_{j=\taus}^BI_j(t^-) - \sum_{j=r}^BS_j(t^-)\right]<0.\label{eq:statement}
\end{align}
\begin{align}
 \ham(t^+)-\lambda_E(t^+)\left[\sum_{j=\taus}^BI_j(t^+) - \sum_{j=r}^BS_j(t^+)\right]<0.\label{eq:statement1}
\end{align}
Furthermore, \eqref{eq:statement} applies for $t=T$.
\end{Lemma}

\paragraph{Proof of Lemma~\ref{lem:tau}}
We start by creating another  control  $\bar{\bu}$ from $\bu$ such that for every $i\geq \taus$, for every $t\leq \tau_i$, $\bar{u}_i(t):=1$, and for every $t>\tau_i$, $\bar{u}_i(t):={u}_i(t)$. We prove by contradiction that  $\tau_i(\bI(0), \bS(0), \bar{\bu})\in \{0, T\}$ for each $i\geq \taus.$ Since $\bar{u}_i\not\equiv u_i$ only in $[0, \tau_i]$ and $I_i(t) = 0$ for {$t \in (0, \tau_i]$} when $\bu$ is used, {the state equations can only differ at a solitary point $t=0$, and therefore} both controls result in the same state evolutions. Thus, for each $i\geq \taus,$ $\tau_i(\bI(0), \bS(0), \bar{\bu})
= \tau_i(\bI(0), \bS(0), \bu)$, and $\tau_i(\bI(0), \bS(0), \bar{\bu})$ may be denoted as $\tau_i$ as well. The lemma therefore follows.

{\hide{Assuming}For the contradiction argument, assume} that the control is $\bar{\bu}$ and that $\tau_i\in(0,T)$ for some $i\geq \taus$. Our proof relies on the fact that  if $\bar{u}_i(t') = 0$ at some $t' \in (0, T)$, then
 $\bar{u}_i(t) = 0$ for $t > t'$, which follows from    Lemma~\ref{lem:reversestepfunction} and the definition of $\bar{\bu}$.
   We break the proof into three parts:

{\bf Case 1:  $i>B-r$}

Here, for $t \in [0, T]$  \eqref{no_rech:I_hi_margin} leads to:
$
I_i(t) = I_i(0) e^{-\beta\int_{0}^{t} \bar{u}_i(t'')\sum_{j=r}^B S_j(t'') \, dt''}.
$
Since $I_i(t) = 0$ for $t \in [0, \tau_i]$, $I_i(0) = 0$. Thus, $I_i(t) = 0$ for all $t \in [0, T].$  So $\tau_i=T$ which contradicts our assumption that $\tau_i\in(0,T).$

{\bf Case 2: $B-\taus<i\leq B-r$}

For $t\in[0,\tau_i]$, since $I_i(t) = 0$ for $t \leq \tau_i$, \eqref{no_rech:I_midhi_margin} becomes $\dot{I}_i = \beta S_{i+r} \sum_{j=\taus}^{B}\bar{u}_j I_j =0$ in this interval. Now, since all elements in $\beta S_{i+r} \sum_{j=\taus}^{B}\bar{u}_j I_j$ are non-negative, we must either have (i) $S_{i+r}(t) = 0$ for some $t \in [0, \tau_i]$, or (ii) for all $\taus \leq k\leq B$, $\bar{u}_k(t) I_k(t)= 0$ for all $t \in[0, \tau_i]$.

(i) In the first case, from two appeals to Theorem~\ref{thm:constraints},  $S_{i+r}(0)=0$ and therefore $S_{i+r}(t)=0$ for all $t \in [0,T]$. So in $[\tau_i,T]$, \eqref{no_rech:I_midhi_margin} becomes $\dot{I}_i = - \beta \bar{u}_i I_i \sum_{j = r}^B S_j$, leading to \begin{align}\label{zeroness}
I_i(t) = I_i(\tau_i) e^{-\beta\int_{\tau_i}^{t} \bar{u}_i(t'')\sum_{j=r}^B S_j(t'')\, dt''}.
\end{align}
Since $I_i(\tau_i) = 0$, $I_i(t) = 0$ for all $t \in [\tau_i,T]$. Therefore $\tau_i=T$
which  contradicts our assumption that $\tau_i\in(0,T).$

(ii) In this case, from \eqref{no_rech:S_gen} to \eqref{no_rech:I_low_margin}, it follows that for all $k\geq 0$, $\dot{I}_k=0$ and $\dot{S}_k=0$ in $[0, \tau_i]$, leading to $\bI(t)=\bI(0)$ and $\bS(t)=\bS(0)$ for $t \in [0, \tau_i]$. Also, since $I_k(t) > 0$ for all $t > \tau_k$,  we know that for all $k\geq \taus$ such that $\tau_k<\tau_i$,  $I_k(t)>0$ for $t \in (\tau_k, \tau_i]$ and therefore $\bar{u}_k(t)=0$ for $t \in (\tau_k, \tau_i]${\hide{, leading to}. This leads to} $\bar{u}_k(t)=0$ for $t\geq \tau_i$ (since  Lemma~\ref{lem:reversestepfunction} {and the definition of $\bar{\bu}$} show that if $\bar{u}_k(t') = 0$ at some $t' \in (\tau_k, T)$, then
 $\bar{u}_k(t) = 0$ for $t > t'$). Especially notice that for all $k\geq \taus$ such that $I_k(0)>0$, $\tau_k=0$ and this would apply. Thus, for each $k$, either $I_k(0) = 0$ or $\bar{u}_k(t) = 0$ for all $t \geq \tau_i$, and hence $I_k(0)\bar{u}_k(t) = 0$ for all $t \geq \tau_i$.   Looking at the interval $[\tau_i, T]$, we prove that $\bS\equiv \bS(0)$ and $\bI \equiv \bI(0)$ constitute solutions to the system of differential equations \eqref{model:no_recharge} in this interval. Replacing these functions {and $\bar{\bu}$} into the RHS of equations \eqref{model:no_recharge}, all terms will be zero (since $I_k(0)\bar{u}_k(t) = 0$ for all $k\geq \taus$, $t \geq \tau_i$), leading to $\dot{I}_k=0$ and $\dot{S}_k=0$ for all $k\geq \taus$, which in turn leads to $(\bS(t), \bI(t)) = (\bS(\tau_i), \bI(\tau_i)) = (\bS(0), \bI(0))$ for all $t \in [\tau_i, T].$  Thus, $\bS\equiv \bS(0)$ and $\bI \equiv \bI(0)$ are the unique solutions  to the system of differential equations \eqref{model:no_recharge} in
  $[0,T]$, wherein uniqueness follows from Theorem~\ref{thm:constraints}. So  $\tau_k \in \{0, T\}$ for these state solutions; a contradiction.

{\bf Case 3: $\taus\leq i\leq B-\taus$}

We prove this case using induction on $i.$ In the induction case, we will consider $i$ such that $\tau_{l} \in \{0, T\}$ for all $l$ such that $i < l \leq B.$
From the arguments for the previous cases, we know that $i=B-\taus$ satisfies the above criterion and therefore constitutes our base case. We only present the proof for the induction case as that for the base case is identical.
 For $t\in[0,\tau_i]$, since $I_i(t) = 0$, \eqref{no_rech:I_gen} becomes $
\dot{I}_i = \beta S_{i+r} \sum_{j=\taus}^{B}\bar{u}_j I_j+ \beta \bar{u}_{i+s} I_{i+s} \sum_{j=r}^B S_j=0.
$
Now, since {both} of these terms\hide{in $\beta S_{i+r} \sum_{j=\taus}^{B}\bar{u}_j I_j+\beta \bar{u}_{i+s} I_{i+s} \sum_{j=r}^B S_j$} are non-negative, each must be equal to zero in $[0, \tau_i].$ 
  As there exists $k\geq r$ such that $S_k(0) > 0$,  there will exist $k\geq r$ such that $S_k(t) > 0$  for all $t\in [0, \tau_i]$ (due to Theorem~\ref{thm:constraints}). Also from the same theorem, we have $S_m(t) \geq 0$ for all $m$. Thus, $\sum_{j=r}^BS_j(t) > 0$ for all $t\in [0, \tau_i]$, and hence the second term is zero contingent on $\bar{u}_{i+s}(t) I_{i+s}(t)=0$ for all $t$ in this interval. So we must either have (I) $S_{i+r}(t)= 0$ {for some $t$} and $\bar{u}_{i+s}(t) I_{i+s}(t)=0$ for all $t$ in this interval, or (II) for all $\taus \leq k\leq B$, $\bar{u}_k(t) I_k(t)= 0$ over this interval. Note that the condition on (II) is exactly the same as in (ii) of Case 2, and following the same argument it may be shown that $\tau_k \in \{0, T\}$ for each $k\geq \taus$ in this case. So we focus on (I):

In (I), again with two appeals to Theorem~\ref{thm:constraints}, we see that $S_{i+r}(0)=0$ and therefore $S_{i+r}(t)=0$ for all $t \in [0,T]$. Thus, for all $t \in [0,T]$, $ \dot{I}_i = -\beta \bar{u}_i I_{i} \sum_{j=r}^{B}S_j
    + \beta \bar{u}_{i+s} I_{i+s} \sum_{j=r}^B S_j$. If $\tau_{i+\taus}<\tau_i$, $I_{i+\taus}(t)>0$ for all $t \in (\tau_{i+\taus}, \tau_i]$  and therefore $\bar{u}_{i+\taus}(t)=0$ for $t \in (\tau_{i+\taus}, \tau_i]$, leading to $\bar{u}_{i+\taus}(t)=0$ for $t\geq \tau_i$. So again, we have \eqref{zeroness} and therefore $\tau_i=T$, a contradiction. If $\tau_{i+\taus}>\tau_i$, on the other hand, for $t \in [\tau_i, \tau_{i+\taus}]$, \eqref{no_rech:I_gen} becomes $ \dot{I}_i = -\beta \bar{u}_i I_{i} \sum_{j=r}^{B}S_j$, again leading to \eqref{zeroness} and thus $I_i(t)=0$ for all $t \in [\tau_i,\tau_{i+s}]$, a contradiction.
Thus, we are left with $\tau_i=\tau_{i+\taus}$. But, since $i < i + \taus \leq B$, $\tau_{i+\taus} \in \{0, T\}.$
Thus, $\tau_i \in \{0, T\}$, which  contradicts our  assumption that $0 < \tau_i < T.$   This completes our proof.

 \subsubsection{Optimal Stopping Time Problem}\label{subsubsec:changes_single_stopping}

{ Using Theorem~\ref{pontryagin1} (with $S_1(\bx^*(t_1),t_1) = R$), the proof differs from the fixed terminal time case only in the arguments used to establish {$\bar{\lambda}_0 =1$ and $\lambda_E > 0$ \hide{in Lemma~\ref{lem:initialinterval}  }in the proof of Lemma~\ref{lem:5} in Appendix~\ref{appendix_lemma_5}. Note that we need separate arguments since the problem is no longer autonomous.}  
Equation~\eqref{vec_neq_zero} along with $\bar{\lambda}_0\geq 0$ leads to $\bar{\lambda}_0=1$, because $\bar{\lambda}_0=0$ would imply: \newline
(i)  $\lambda_{i}(T)= \rho_i(T)=0, \quad \forall i=0,\ldots,B$, \newline
(ii) $\ham(T)=\lambda_E(T)\sum_{i=\taus}^BI_i(T)=\bar{\lambda}_0f'(T)=0$.
The first equality in (ii) comes from replacing $\lambda_i(T)=\rho_i(T)=0$ for all $i$ into~\eqref{define:Hamiltonian:Gen}, and the second from~\eqref{H(T)equals}. Now, there exists a $j\geq \taus$ such that $I_j(0)>0$, and due to Theorem~\ref{thm:constraints},  $I_j(T) > 0$, and  $I_m(T) \geq 0$ for all $m$. Thus, $\sum_{i=\taus}^BI_i(T) > 0$,
leading to $\lambda_E(T)=0$. This, combined with $\bar{\lambda}_0=0$ and (i), {contradicts \eqref{vec_neq_zero}} at $t=T$.}

Thus, henceforth we consider $\bar{\lambda}_0=1$. {As in Lemma~\ref{lem:initialinterval}, it can be shown that $u_i(T)=0$ for all $i\geq \taus$.
So we again have $
 H(T)=\lambda_E(T)\sum_{i=\taus}^BI_i(T),
$
and, from~\eqref{H(T)equals}, $ 
f'(T)=\lambda_E(T)\sum_{i=\taus}^BI_i(T)$.
Since $
f'(T)>0$ and $\sum_{i=\taus}^BI_i(T) > 0$, $\lambda_E(T)>0$.
The rest of the proof is identical to that for the fixed terminal time case.

\subsection{Proof of Theorem~\ref{Thm:Order}}
\label{sec:proof2}
We present the proof without explicitly mentioning which version of the optimal control problem
(fixed terminal time or optimal stopping time) we are considering since the proof is identical.
We will use Lemma~\ref{lem:varphi_zero_X}, \eqref{eq:varphi_simpler:Gen}, \eqref{optimal_u_i:Gen}, 
 and the values of $\lambda_i(T), \rho_i(T)$ from \eqref{co_st_finals:Gen} which hold for both versions.

We will prove this theorem for an optimal control $\bu$ such that $u_i \equiv 0$ for all  $i \not\in \mathcal{Z}(\bu).$ It is sufficient to consider only such optimal controls because for any optimal control $\tilde{\bu}$ we can  construct a control $\bu$ such that $u_i(t):= 0$ for $i \not\in \mathcal{Z}(\tilde{\bu})$ and $u_i:= \tilde{u}_i$ for $i \in  \mathcal{Z}(\tilde{\bu}).$  Since $\bu$
leads to the same {\hide {process}state evolutions} as $\tilde{\bu}$,   $\bu$ is  optimal,  $\mathcal{Z}(\tilde{\bu})
= \mathcal{Z}(\bu)$, and both controls have identical threshold times for $i \in  \mathcal{Z}(\tilde{\bu}) = \mathcal{Z}(\bu).$ The theorem therefore follows for  $\tilde{\bu}$ {\hide{from that}if it is proven} for  $\bu.$

The result clearly holds if $\bu \equiv 0$ as then ${t_i = t_j=0}$ for all $i, j \in \mathcal{Z}(\bu).$
We therefore assume that $\bu \not\equiv 0$. It suffices to show that
if $\varphi_i(t)=0$ for some $t>0$ and for  $i \in \mathcal{Z}(\bu)$, we have
$\varphi_k(t)\leq 0$ for any $ k< i$ {where we have }$k \in \mathcal{Z}(\bu)$. From the definition of  $\mathcal{Z}(\bu)$,
$\tau_i = \tau_k = 0.$ Then,
from Lemma~\ref{lem:varphi_zero_X} and~\eqref{optimal_u_i:Gen},
the threshold time
for $u_k$
will precede that of $u_i.$

To prove the above, we examine two cases: (1) $\bar{\lambda}_0=0$ and (2) $\bar{\lambda}_0=1$. In case (1), $\rho_i(T)=\lambda_i(T)=0$ for all $i$, leading to $\varphi_i(T)=0$ for all $i\geq \taus$ from \eqref{eq:varphi_simpler:Gen}. From  Lemma~\ref{lem:varphi_zero_X}, this means that $\varphi_i(t)>0$ for all $0 < t<T$ and all $i\in \mathcal{Z}(\bu)$ (note that $\tau_i = 0$ if $i\in \mathcal{Z}(\bu)$). Therefore, from \eqref{optimal_u_i:Gen}, $u_i(t) = 1$ for all $t \in (0, T)$; thus $t_i=T$ for all $i \in \mathcal{Z}(\bu)$. Thus, henceforth we focus on the case where $\bar{\lambda}_0=1$.


Consider {an} $i \in \mathcal{Z}(\bu)$ {and} a time $\sigma_i>0$ such that $\varphi_i(\sigma_i)=0$.
From~\eqref{eq:varphi_simpler:Gen} we have: $
 \varphi_i(\sigma_i)=\beta I_i\left(\sum_{j=r}^B\left(-\lambda_j+\rho_{j-r}+\rho_{i-\taus}-\rho_i\right)S_j\right)\bigg|_{t=\sigma_i} = 0.
$
Note that $I_i(\sigma_i)>0$ (since $i \in \mathcal{Z}(\bu)$, $\sigma_i > 0$); thus, at $t=\sigma_i$,
\begin{align*}
 \sum_{j=r}^B\left(-\lambda_j+\rho_{j-r}\right)S_j = - \sum_{j=r}^B\left(\rho_{i-\taus}-\rho_i\right)S_j.
\end{align*}
Using the above and~\eqref{eq:varphi_simpler:Gen},
it turns out that for all $k \in \mathcal{Z}(\bu)$,
$
\varphi_k(\sigma_i)= {\hide{\beta I_k  \left(-\rho_{i-\taus}+\rho_i+\rho_{k-\taus}-\rho_k\right)
\sum_{j=r}^BS_j\notag\\ &:=} \beta I_k\psi_{i,k}(\sigma_i)\sum_{j=r}^BS_j,}
$
where $\psi_{i,k},$  for $\taus\leq k< i$, is defined as:
\begin{align*}
\psi_{i,k}(\sigma_i):=
-\rho_{i-\taus}+\rho_i+\rho_{k-\taus}-\rho_k.
\end{align*}
 We know that {$\sum_{j=r}^BS_j(\sigma_i)\geq 0$,  $I_k (\sigma_i)> 0$}  (from Theorem~\ref{thm:constraints}). The theorem now follows from the following lemma:
\begin{Lemma}\label{lem:psi_at_sigma}
 For any {$k< i$ such that $i, k \in \mathcal{Z}(\bu)$} and for $\sigma_i>0$ such that $\varphi_i(\sigma_i)=0$, we have $\psi_{i,k}(\sigma_i)\leq 0$.
\end{Lemma}
{\hide{The proof of the above lemma does not follow from standard optimal control theory, and is therefore one of our theoretical contributions.}} 

\begin{proof}
\hide{The case of $\taus\leq k<2\taus$ follows similarly.}
At $t=T$, following~\eqref{co_st_finals:Gen}, we have:
\begin{align*}
\psi_{i,k}(T)=&-\rho_{i-\taus}(T)+\rho_i(T)+\rho_{k-\taus}(T)-\rho_k(T) \\=&
[a_{i-\taus}-a_i-(a_{k-\taus}-a_k)],
\end{align*}
which due to the  properties assumed for $a_i$ ($a_i$ decreasing and strictly convex in $i$), yields $\psi_{i,k}(T)< 0$.
This also holds on a sub-interval of nonzero length that extends to $t=T$, owing to the continuity of $\psi_{i,k}$.
We now prove the lemma by contradiction: 
going back in time from $t=T$ towards $t=\sigma_i$, suppose a $\psi_{i,k}$ becomes non-negative at time $\bar\sigma>\sigma_i$ for some {$ k< i$, $k \in \mathcal{Z}(\bu)$}. That is, for at least one such $k$ we have:
{\begin{align}
&(-\rho_{i-\taus}+\rho_i+\rho_{l-\taus}-\rho_l)< 0\notag
 \\&\qquad\qquad\text{ $\forall  l<i,~ l \in \mathcal{Z}(\bu)$ , $\forall t$} ~~\sigma_i<\bar\sigma< t\leq T; \label{rho_order}
\end{align}
and at $t=\bar\sigma$,
\begin{align}
\begin{cases}\label{rho_sigma}
(-\rho_{i-\taus}+\rho_i+\rho_{k-\taus}-\rho_k)=0\\
(-\rho_{i-\taus}+\rho_i+\rho_{l-\taus}-\rho_{l})\leq0, \forall \ l<i, l \in \mathcal{Z}(\bu).
\end{cases}
\hide{ \text{ at }
t=\bar\sigma.}
\end{align}}
We show that the time derivative of $\psi_{i,k}$ 
is non-negative over the points of continuity of the controls in the interval $[\bar\sigma,T]$. Note that this, plus the continuity of $\psi_{i,k}$, leads to a contradiction with the existence of $\bar\sigma$ and hence proves the lemma, since:
 $\psi_{i,k}(\bar\sigma)=\psi_{i,k}(T)-\int_{t=\bar\sigma}^T\!\dot{\psi}_{i,k}(\nu)\,d\nu\leq \psi_{i,k}(T)< 0$.
We now investigate $\dot{\psi}_{i,k}$ over the points of continuity of the controls in $[\bar\sigma,T]$.\footnote{{Note that since $i, k \in  \mathcal{Z}(\bu)$, $I_i(t)>0$ and $I_k(t)>0$ for all $t>0$.}}
{\hide{For {$\taus \leq k<2\taus \leq i$ such that $k \in \mathcal{Z}(\bu)$} it follows that:
\begin{align}\label{psi_dot2}
 \dot{\psi}_{i,k}=&
\beta u_{i-\taus} (\sum_{m=r}^B\left[-\lambda_m+\rho_{m-r}+\rho_{i-2\taus}-\rho_{i-\taus}\right]S_m)\notag\\&
+\lambda_E-\frac{\varphi_{i}u_{i}}{I_{i}}+\frac{\varphi_{k}u_k}{I_k} 
\end{align}
and f}For $\taus \leq k < i<2\taus$ such that $k \in \mathcal{Z}(\bu)$:
\begin{align}\label{psi_dot3}
 \dot{\psi}_{i,k}=&-\frac{\varphi_{i}u_{i}}{I_{i}}+\frac{\varphi_{k}u_k}{I_k},
\end{align}
and for {$\taus \leq k<2\taus \leq i$ such that $k \in \mathcal{Z}(\bu)$} it follows that:
\begin{align}\label{psi_dot2}
 \dot{\psi}_{i,k}=&
\beta u_{i-\taus} (\sum_{m=r}^B\left[-\lambda_m+\rho_{m-r}+\rho_{i-2\taus}-\rho_{i-\taus}\right]S_m)\notag\\&
+\lambda_E-\frac{\varphi_{i}u_{i}}{I_{i}}+\frac{\varphi_{k}u_k}{I_k}. 
\end{align}
}
The RHS of (\ref{psi_dot3}-\ref{psi_dot2}) is non-negative because:
\begin{itemize}
\item[(A)] $\dfrac{\varphi_{k}u_k}{I_k}$ is non-negative due to~\eqref{optimal_u_i:Gen} for all $k\geq \taus$,
\item[(B)] $u_{i-\taus} (\sum_{m=r}^B\left[-\lambda_m+\rho_{m-r}+\rho_{i-2\taus}-\rho_{i-\taus}\right]S_m)$
 is  non-negative for $i\geq 2\taus$. To see this, note that for $i$ such that $I_{i-s}(t)> 0$ for $t>0$  this term is equal to $\dfrac{\varphi_{i-s}u_{i-s}}{ I_{i-s}}$ which is non-negative, again as imposed by the optimizations in~\eqref{optimal_u_i:Gen}; else {$(i-s) \not\in \mathcal{Z}(\bu)$} and $ u_{i-s}\equiv0$;
\item[(C)] $\varphi_{i}(t)u_{i}(t)=0$ for $t \geq \sigma_i$. To see this note that  $\varphi_{i}(\sigma_i) = 0$.
 For $t > \sigma_i$,  from Lemma~\ref{lem:varphi_zero_X}, we have $\varphi_i(t) < 0$, which together with \eqref{optimal_u_i:Gen} leads to $u_{i}(t)=0$,
  \item[(D)] $\lambda_E= \lambda_E (T)>0$, as established after \eqref{eq:the_first_expr} for the
  fixed terminal time problem  and in \S\ref{subsubsec:changes_single_stopping} for the optimal stopping time problem. 
 \end{itemize}

 For $i>k\geq 2\taus$ we have:
\begin{align}\label{psi_dot}
 \dot{\psi}_{i,k}&=\beta u_{i-\taus} (\sum_{m=r}^B\left[-\lambda_m+\rho_{m-r}+\rho_{i-2\taus}-\rho_{i-\taus}\right]S_m) \notag \\
&\hspace{-0.1in}- \beta u_{k-\taus} (\sum_{m=r}^B\left[-\lambda_m+\rho_{m-r}+\rho_{k-2\taus}
-\rho_{k-\taus}\right]S_m)\notag\\
&-\frac{\varphi_{i}u_{i}}{I_{i}}+\frac{\varphi_{k}u_k}{I_k}\notag\\
&\hspace{-0.25in}\geq 
- \beta u_{k-\taus} (\sum_{m=r}^B\left[-\lambda_m+\rho_{m-r}+\rho_{k-2\taus} -\rho_{k-\taus}\right]S_m).
\end{align}
The above inequality follows from (A), (B), (C) above.
Now we show that the RHS in the last line  is zero over the interval of
$[\bar\sigma,T]$, completing the argument. If  $k-s \not\in \mathcal{Z}(\bu)$,   then $u_{k-s}\equiv0$. Else, $I
_{k-\taus}(t)>0$ for all $t>0$, and the RHS   equals $\dfrac{\varphi_{k-\taus}u_{k-\taus}}{I_{k-\taus}}$. We now show that $\varphi_{k-\taus}(t)\leq0$ for all $t \in [\bar\sigma,T]$; thus  \eqref{optimal_u_i:Gen} leads to $\varphi_{k-\taus}(t)u_{k-\taus}(t)=0$, for all $t \in [\bar\sigma,T]$. The result follows.



From~\eqref{eq:varphi_simpler:Gen}, we have:
\begin{align*}
\begin{cases}
  \varphi_i&= \beta I_i\left(\sum_{j=r}^B\left(-\lambda_j+\rho_{j-r}+\rho_{i-\taus}-\rho_i\right)S_j\right) \\
\varphi_{k-\taus}&=
\beta I_{k-\taus}\left(\sum_{j=r}^B\left(-\lambda_j+\rho_{j-r}+\rho_{k-2\taus}-\rho_{k-\taus}\right)S_j\right)
\end{cases}
\end{align*}
Now, since $I_i(t)>0$ for $t>0$, $\varphi_i(t)\leq 0$ leads to:
$
 \sum_{j=r}^B\left(-\lambda_j+\rho_{j-r}+\rho_{i-\taus}-\rho_i\right)S_j\leq0.
$
From~(\ref{rho_order}, \ref {rho_sigma}) and for $k'=k-\taus<i$, we have $\rho_{k-2\taus}-\rho_{k-\taus}\leq  \rho_{i-\taus}-\rho_i$ over the interval of $[\bar\sigma,T]$. Hence we now have:
$
\sum_{j=r}^B\left(-\lambda_j+\rho_{j-r}+\rho_{k-2\taus}-\rho_{k-\taus}\right)S_j\leq0,
$
which together with $I_{k-\taus}(t)\geq0$ for $t>0$ results in $\varphi_{k-\taus}(t)\leq0$.

This concludes the lemma, and hence the theorem.
\end{proof}

\section{Numerical Investigations}\label{sec:Numericals}

\hide{We start by assessing the validity of our modeling assumptions.} \hide{We have analyzed the system in the mean-field deterministic regime which models state evolution using a system of differential equations  \eqref{model:no_recharge}. Such models have been shown to be acceptable approximations both analytically and empirically for large and fast-moving mobile wireless networks~\cite{khouzani2012optimal}. \hide{In our technical report~\cite{techreport}, we independently investigate the accuracy of our model using simulations, and  use this model for  subsequent evaluations. }Next, we compare our optimal policy to the}{Numerous heuristic policies have been proposed for message passing in DTNs in prior literature~\cite{vahdat2000epidemic, zhang2007performance,lindgren2003probabilistic, de2009nectar,banerjee2010design,spyropoulos2005spray,wang2006dft, lu2010energy,nelson2009encounter,de2010optimal,singh2011optimal, singh2010delay,balasubramanian2007dtn,altman2010optimal,neglia2006optimal}. Many of these heuristics are simpler to implement than our optimal control as they employ controls that either do not depend on residual energy levels or do not change with time. 
{ We start by experimentally validating the mean-field deterministic model we used (\S\ref{subsec:model}) and quantifying the benefit of our optimal policy relative to some of these heuristics (\S\ref{subsec:heuristics}).}\hide{, we quantify the \hide{relative }benefit of our proposed optimal policy over {\hide{that of }}these heuristics\hide{ (and some that we propose)}.}
{Next, we investigate the sensitivity of our optimal control to errors in clock synchronization {and residual energy {determination\hide{estimation}}} among nodes (\S\ref{subsec:robustness}).}
Finally, in \S\ref{subsec:lifetime}, we investigate the sending of multiple messages over successive time intervals empirically,  and assess the performance of a natural generalization of our policy (which is optimal for the transmission of a {\em single} message) relative to that of the mentioned heuristics.}

We focus on the fixed terminal time problem and  derive the optimal controls using the GPOPS software \cite{rao2010algorithm, benson2006direct, garg2011direct, garg2010unified, garg2011pseudospectral} with INTLAB \cite{Ru99a}.
 Unless otherwise stated, our  system used parameters: $B=5$, $\taus=2$, and $r=1$ (note that $\taus\geq r$, as demanded by our system model), $\mathbf{S_0}=(0,0,0,0.3,0.3,0.35)$, and
 $a_i= (B-i)^2$.  Note that $\beta T$ denotes the average number of contacts of each node in the system in the time interval $[0,T]$. 
Thus, as expected we observed that changing $\beta$ and $T$ had very similar effects on the costs and the drop-off points of the optimal controls. We further assumed that $\beta=\beta_0$ (i.e., the rate of contact between any two nodes is the same as the rate of contact of the destination and any given node). {\hide{We assessed each policy based on its {``Energy Cost''} $\sum_{i=\taus}^B a_i (S_i(T)+ I_i(T))$. Since this value also depends on the initial penalty function value,  $\sum_{i=\taus}^B a_i (S_i(0)+ I_i(0))$, which is the same for all policies,}}We compared policies based on the difference between $\sum_{i=\taus}^B a_i (S_i(T)+ I_i(T))$ and $\sum_{i=\taus}^B a_i (S_i(0)+ I_i(0))$ (which, as the initial penalty function value, is the same for all policies)\hide{they attain} for each policy,  which we call the ``Unbiased  Energy Cost''.

{\subsection{Validation of the mean-field deterministic model}
\label{subsec:model}
We noted in \S\ref{subsec:model_dynamics}  that assuming exponential  contact among  nodes leads to the system dynamics~\eqref{model:no_recharge} (the mean-field deterministic regime) in the limit that the number of nodes, $N$, approach $\infty$.
 We therefore assess  the applicability of \eqref{model:no_recharge} for exponential contact processes  and large, but finite $N$ (\S\ref{subsubsec:exponential}).  
 Subsequently we assess the validity of \eqref{model:no_recharge}  for a specific truncated power-law contact process that was experimentally observed for human mobility at INFOCOM 2005~\cite{hui2005pocket} (\S\ref{subsec:model2}). Under this model, nodes do not mix homogenously, as those that have met in the past are more likely to meet in the future, {and their convergence to ODEs like ours has not been established.}

 For each contact process, we  simulated 100 runs of the evolution of the states  with  forwarding probabilities provided by the optimal control for the fixed terminal time problem and {state equations} \eqref{model:no_recharge}.    We compared the average state evolutions and {unbiased energy costs} of these cases with those obtained from  \eqref{model:no_recharge}  under the same control.  We describe the results below.

\subsubsection{ Exponential Contact Process}
\label{subsubsec:exponential}
{\hide{With}For a system with $N=160$ nodes, $\mathbf{I_0}=(0,0,0,0.0125,\allowbreak0.0125,0.025)$,   $\beta=2$, and $T=5$,} leading to  an average of 10 meetings per node, Fig.~\ref{fig:comparison_costs} and Fig.~\ref{fig:comparison_states_exp} reveal that the results obtained from the simulation of the exponential contact process and  \eqref{model:no_recharge} are similar, as expected.

 \begin{figure}[htb]
 \centering
\includegraphics[scale=0.45]{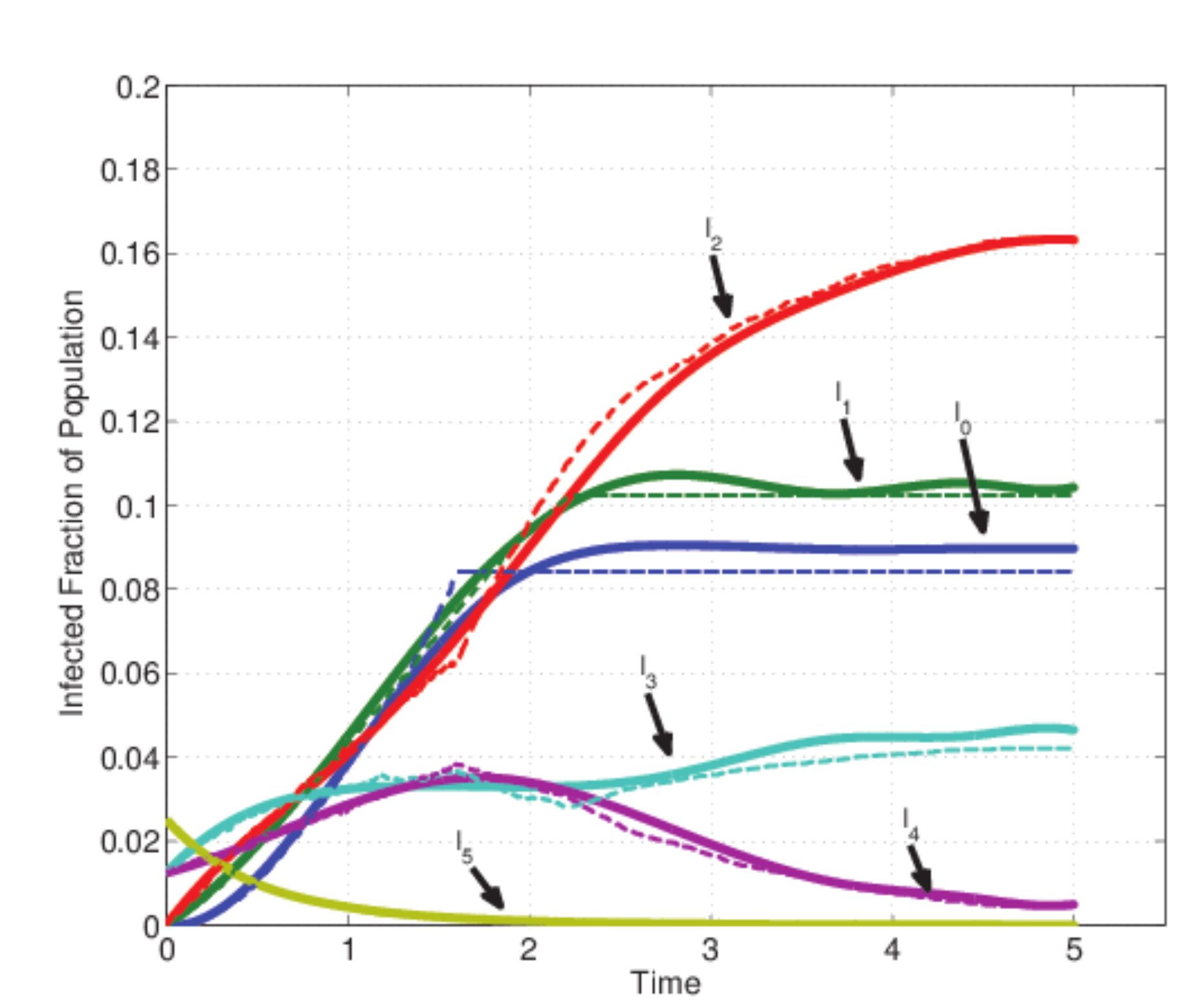}
\caption{{\scriptsize Comparison of the state processes for the mean-field {deterministic} regime (dashed lines) and  simulated exponential contact process. We consider a mandated probability of delivery of 80\%.}}
\label{fig:comparison_costs}
\end{figure}
 \begin{figure}[htb]
 \centering
\includegraphics[scale=0.45]{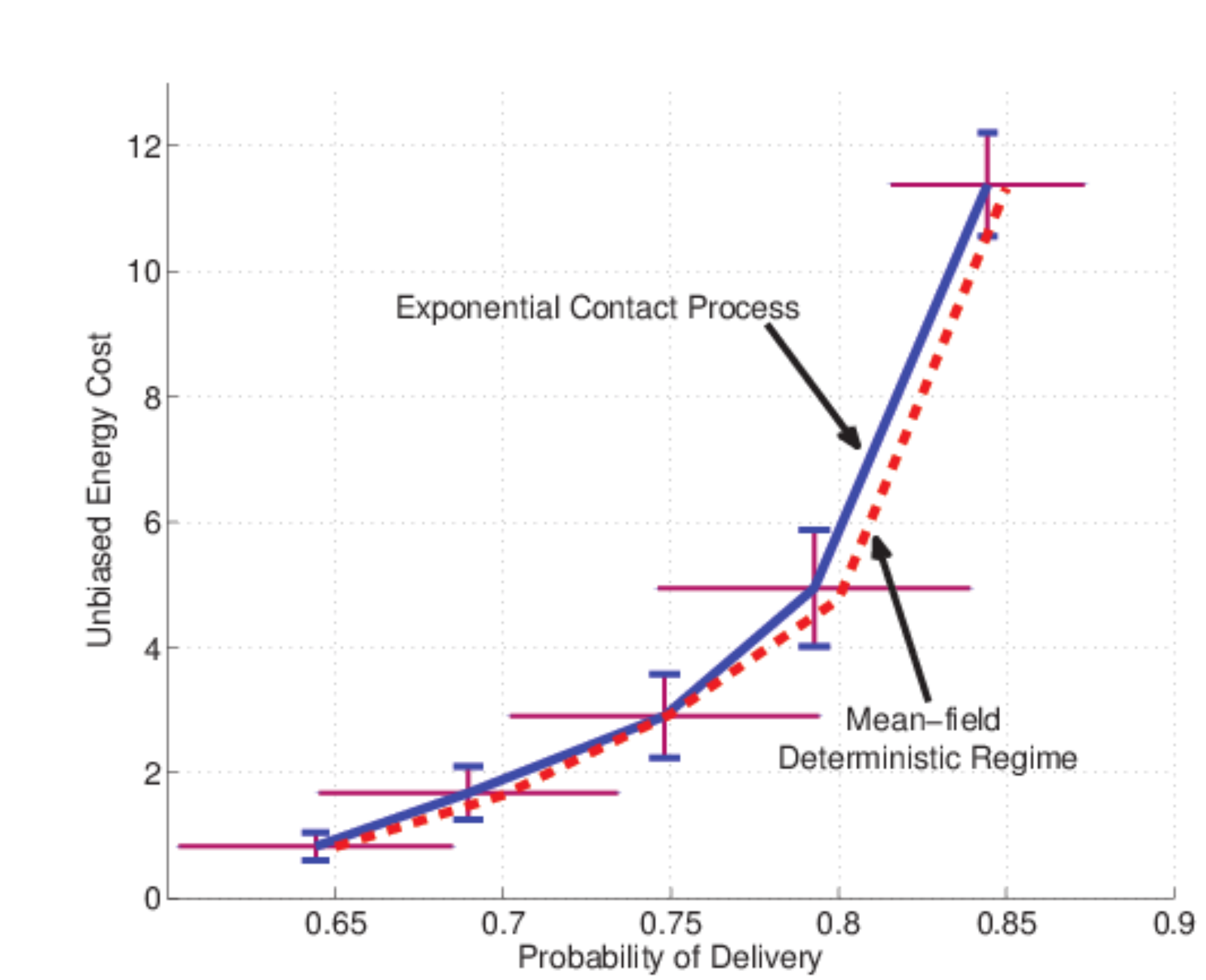}
\caption{{\scriptsize Comparing the costs of the mean-field  deterministic   regime (dashed line) and simulated exponential process as a function of the mandated probability of delivery. The error bars represent the standard deviations of the statistical simulations.}}
\label{fig:comparison_states_exp}
\end{figure}
{
\subsubsection{Truncated Power Law Contact Process}\label{subsec:model2}
We consider the truncated power-law contact process observed in~\cite{hui2005pocket} for a network with $N=41$ nodes and $\alpha=0.4$. The power-law process  was truncated in that the contact times are restricted to be  between $2$ minutes and $24$ hours. We use $\beta = 4.46$ in our differential equations \eqref{model:no_recharge}  so that $1/\beta$ equals the
expected inter-contact time between any pair
of nodes under this distribution. Also, $\mathbf{I_0}=(0,0,0,0,0.025,0.025)$. Even though $N$ is small and the contact process is not memoryless, Fig.~\ref{fig:comparison_states_pwr} shows that the states derived from this simulation and \eqref{model:no_recharge} follow the same trends, but there is a gap, which is to be expected because this contact model does not have the homogeneity of the exponential case, and the number of nodes is small ($N=41$, since the experimental data in~\cite{hui2005pocket} was obtained for this $N$). Fig.~\ref{fig:comparison_pwr} show that the costs in this model are, however, quite close to those derived from our equations, suggesting the robustness of energy cost to the change in contact process. \hide{obtained from the truncated power law process  are reasonably close to that of the deterministic model and follow the same general trends.}

 \begin{figure}[htb]
 \centering
\includegraphics[scale=0.45]{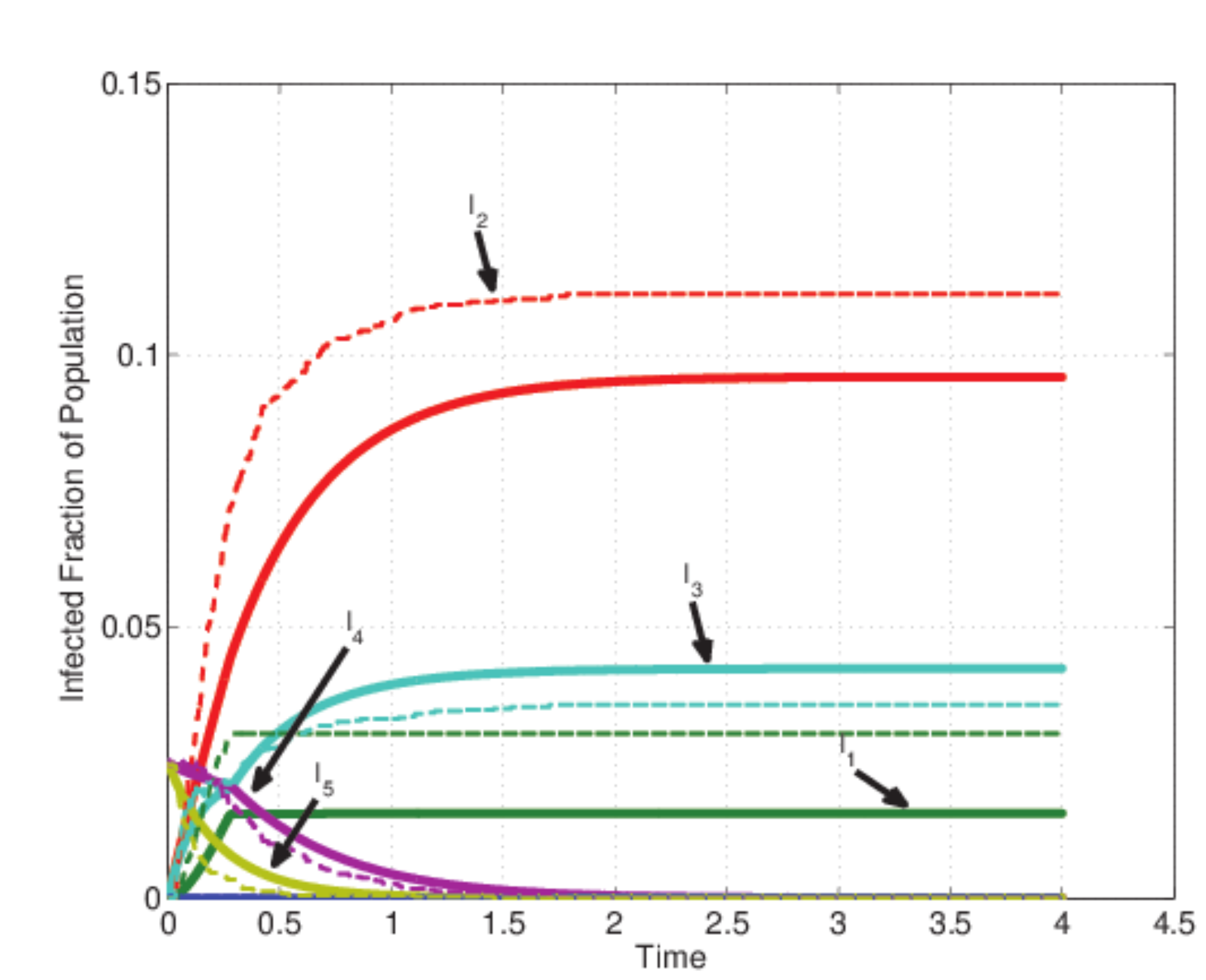}
\caption{{\scriptsize Comparison of the evolution of the infection in a mean-field deterministic regime (dashed lines) and the power-law contact process
observed in \cite{hui2005pocket}. We use a mandated probability of delivery of 90\%.
 }}
\label{fig:comparison_states_pwr}
\end{figure}

 \begin{figure}[htb]
\centering
\includegraphics[scale=0.45]{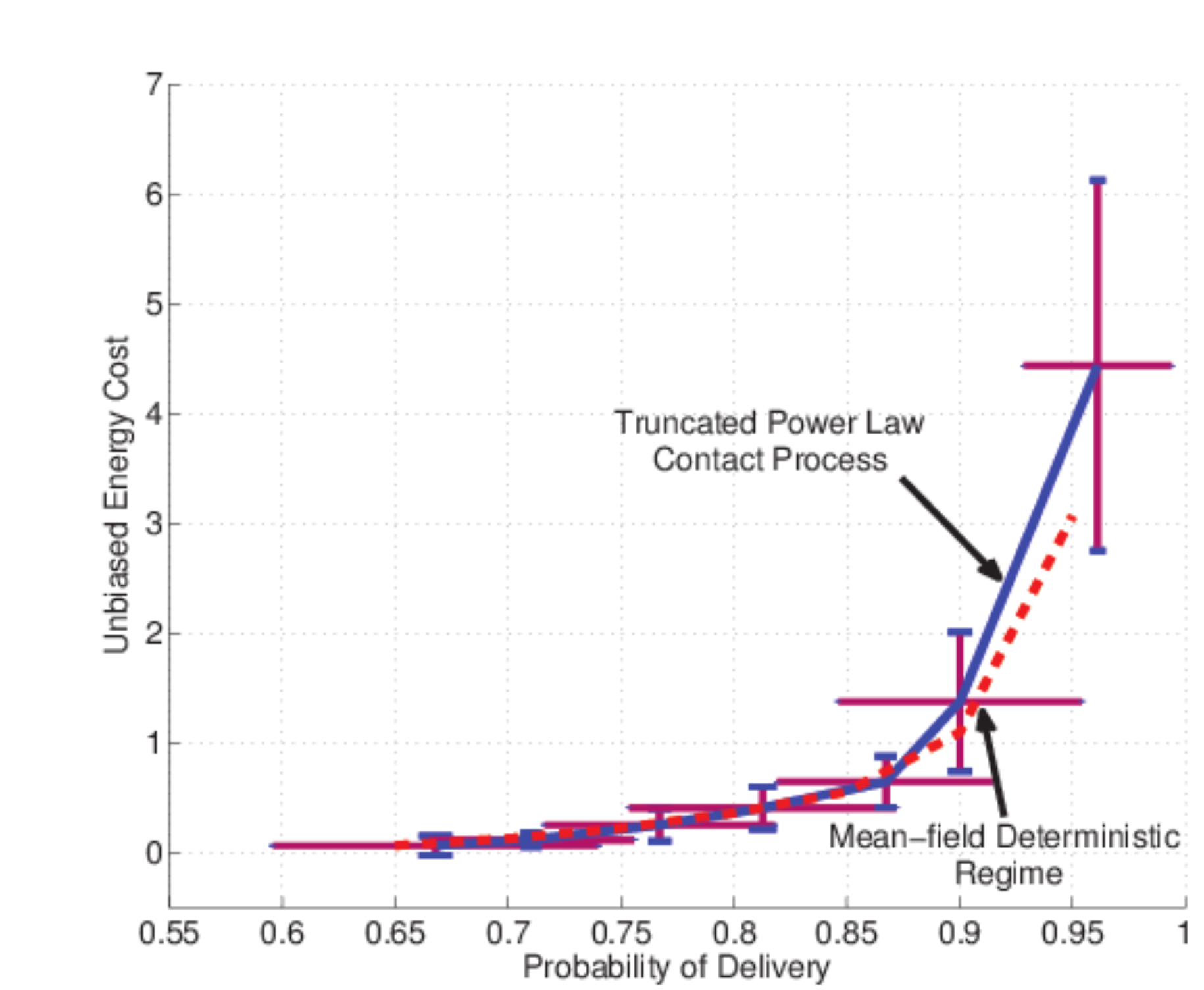}
\caption{{\scriptsize Here, the cost under the mean-field deterministic regime (dashed line) is compared to that of a truncated power law contact process (solid line) as a function of the mandated probability of delivery. The error bars represent the standard deviations of the costs and probabilities of delivery.}}
\label{fig:comparison_pwr}
\end{figure}
}}
\subsection{Performance advantage of optimal control over heuristics}
\label{subsec:heuristics}

\subsubsection{Description of Heuristics}
We propose two classes of heuristic policies, and describe sub-classes that correspond to policies in prior literature. In all classes and sub-classes, we define the {\em best} policy to be that which minimizes the unbiased energy  cost{\hide{  (i.e., $\sum_{i=\taus}^B \left(S_i(T) + I_i(T)\right)$}} subject to satisfying the throughput constraint \eqref{path_constraint}.

{\bf I}- Static Across Energy Levels: Policies that choose a one jump (from a fixed value in $[0,1]$ to zero) control that is the same for all energy levels. In these policies, nodes do not need to know their residual energy level. {\hide{This information is readily available in most current hardware but might be less accessible and less exact in legacy versions.}}The best policy in this class is selected through a search over the range $[0,T]\times[0,1]$, which is less than that of the optimal control ($[0,T]^{B-s+1}$).

{\bf II}- Static Across Time: Policies that force all controls to be at a fixed value (potentially different for each energy level) throughout $[0,T]$. These policies are inherently robust to errors in clock synchronization,  and the best policy in this class can be determined through a search over  the range $[0,1]^{B-s+1}$, which is similar to that of the optimal control.

Policies where controls depend on residual energy levels, e.g., those in (II), have not been proposed in existing literature.
Several sub-classes of (I) have been proposed, however:\footnote{Sub-classes inherit constraints of classes from which they are descended.}

\indent 1) Probability Threshold (also known as optimized flooding): Policies whose controls  drop from 1 to 0 when the probability of message delivery in $[0,T]$ surpasses a certain threshold (e.g., \cite{wang2006dft}).

\indent 2) Infection Threshold: Policies whose controls  drop from 1 to 0 when the total number of infected nodes with enough energy to transfer the message to the destination surpasses a certain threshold (e.g., \cite{neglia2006optimal}).

 \indent 3) Static Across Time and Energy Levels: Policies that force all energy levels to choose the same fixed control (between 0 and 1) throughout $ [0,T]$ (e.g., \cite{neglia2006optimal}).

\indent 4) One Control (also known as flooding, epidemic routing): The single policy that sets all controls to one. (Originally in \cite{vahdat2000epidemic}, also in \cite{de2010optimal} and \cite{wang2006dft}.)

\indent 5) Zero Control (also known as Spray and Wait, two-hop transmission, direct transmission): The single policy that sets all controls to zero. (Originally in \cite{spyropoulos2005spray}, also in \cite{de2010optimal} and \cite{wang2006dft}.)

The best policy in the Probability and Infection  Threshold classes can be determined through a search over
 $[0,T]$, and that in the Static Across Time and Energy Levels class through a search over $[0,1]$. {\hide{Note that the Zero and One Control policies and also those  in the Static Across Time and Energy Levels class  are inherently robust to errors in clock synchronization and do not  need  knowledge of the residual energy levels.  }}However, the Zero Control policy fails to attain  the mandated probability of delivery in settings that we consider (small to moderate values of  initial infection  and $T$), and is thus excluded from Fig.~\ref{fig:perfor2} and \ref{fig:diversity} presented below.

\subsubsection{Relative Performance}\label{subsec:performance}

In Fig.~\ref{fig:perfor2}, the costs associated with energy consumption for the optimal policy and also  the best policies in each of  the proposed classes are compared as $\beta$ is varied. We  use the name of the class/sub-class to refer to the best policy in that class/sub-class. The mandated probability of delivery is $90\%$, while  $\mathbf{I_0}=(0,0,0.0125,0.0125,0.0125,0.0125)$.
As the number of contacts increases, forwarding the message at every available opportunity becomes less desirable as it leads to massive energy consumption. The ``One Control'' policy, therefore, acts as a battery depletion attack on the nodes, using up all of their energy reserves and leading to significantly higher cost (over 30\% worse than the second worst heuristic), and therefore it is left out of the figure for illustrative purposes.\hide{ This is the main flaw of simple epidemic routing~\cite{vahdat2000epidemic} that spurred the development of the other heuristics in the time since its proposition. }\hide{It is interesting to note that as the number of contacts increases, the One Control policy is the only one that shows a deterioration in performance, as an increased number of contacts translates to an increase in energy use for this heuristic with no additional gain. The other heuristics can  curb their transmissions as required; thus, intuitively,  {as an increase in $\beta_0=\beta$ leads to more contacts with the destination and thus less need for forwarding the message to intermediate nodes in the first place}, their energy consumption decreases.
\begin{figure}[htb]
 \centering
\includegraphics[scale=0.45]{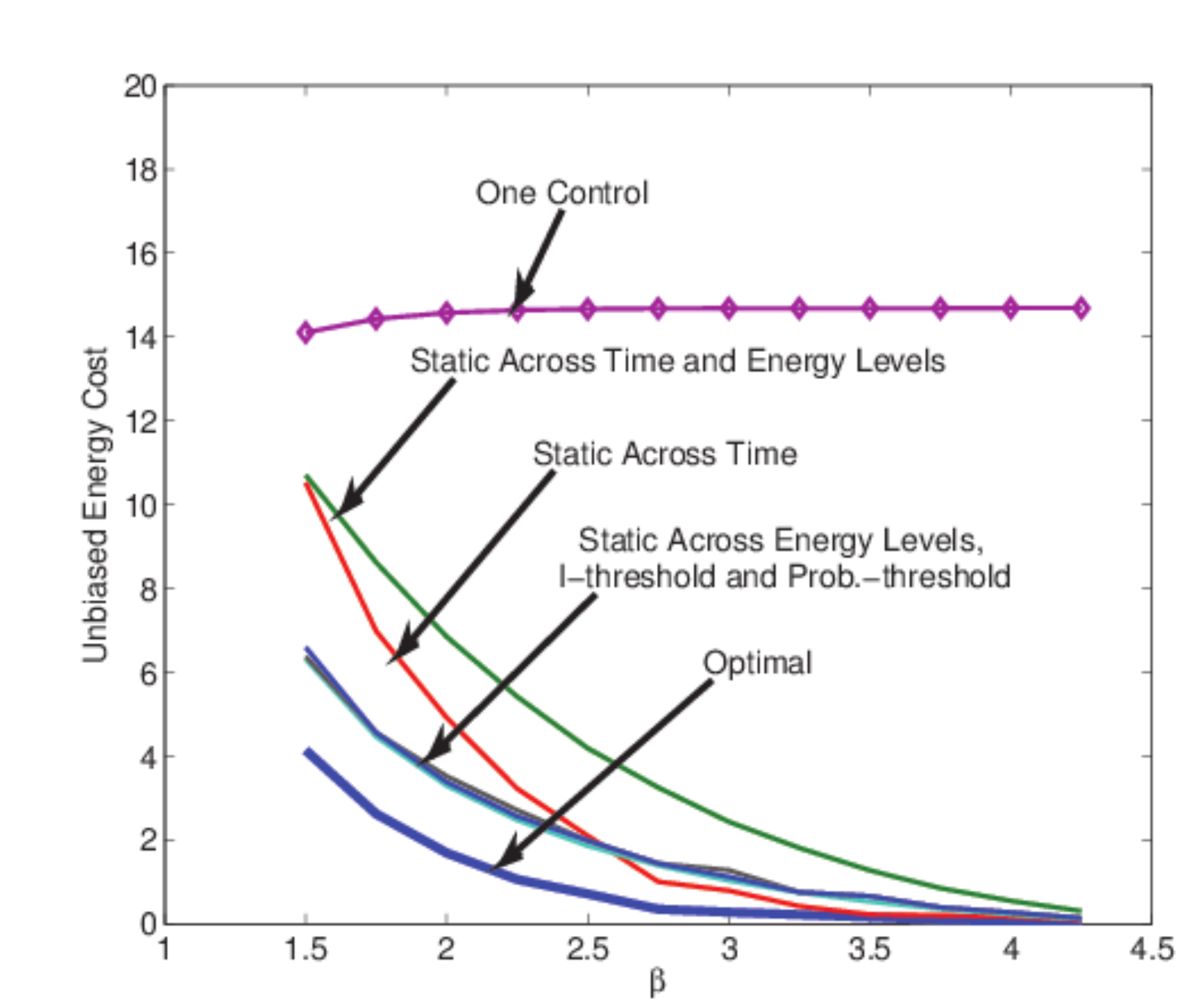}
\caption{{\scriptsize Performance of the optimal and heuristic controls.  The performances of the ``Static Across Energy Levels'', ``Infection Threshold'', and ``Probability Threshold'' policies are very close, and they are indicated with a single arrow.}}
\label{fig:perfor1}
\end{figure}

To get a better view of the performance of the other competing policies, we remove ``One Control'' and compare the rest (Fig.~\ref{fig:perfor2}). Here,} We see that the optimal policy significantly outperforms the best of the rest of the heuristic for low and moderate values of $\beta$ (for $\beta \leq 2.5$), e.g.,  the performance difference is  $50\%$ for $\beta \approx 2$. We also see that the Static Across Energy Levels
 and Static in Time heuristics respectively outperform  all  other heuristics for low and high values of $\beta$. As contacts ($\beta T$) increase, the flexibility to adapt the control in accordance with the residual energy of the nodes provided by Static in Time  turns out to be beneficial, as the mandated probability of delivery can be achieved by utilizing higher energy nodes. {{In fact, Static in Time
 performs close to the optimal for large values of $\beta$.}} In summary, the improvement in performance attained by the optimal control over simpler heuristics justifies its utilization of time-dependent and residual-energy-dependent decisions except for relatively large values of $\beta$ where there is less need to spread the message due to more frequent meetings with the destination. In this case, near-optimal performance can be achieved by choosing controls based only on residual energy and not  time, as is the case for Static in Time. {{Such choices  may be used instead of the optimal policy for more robustness to clock synchronization errors{, an issue we visit next}.}}
\begin{figure}[htb]
 \centering
\includegraphics[scale=0.45]{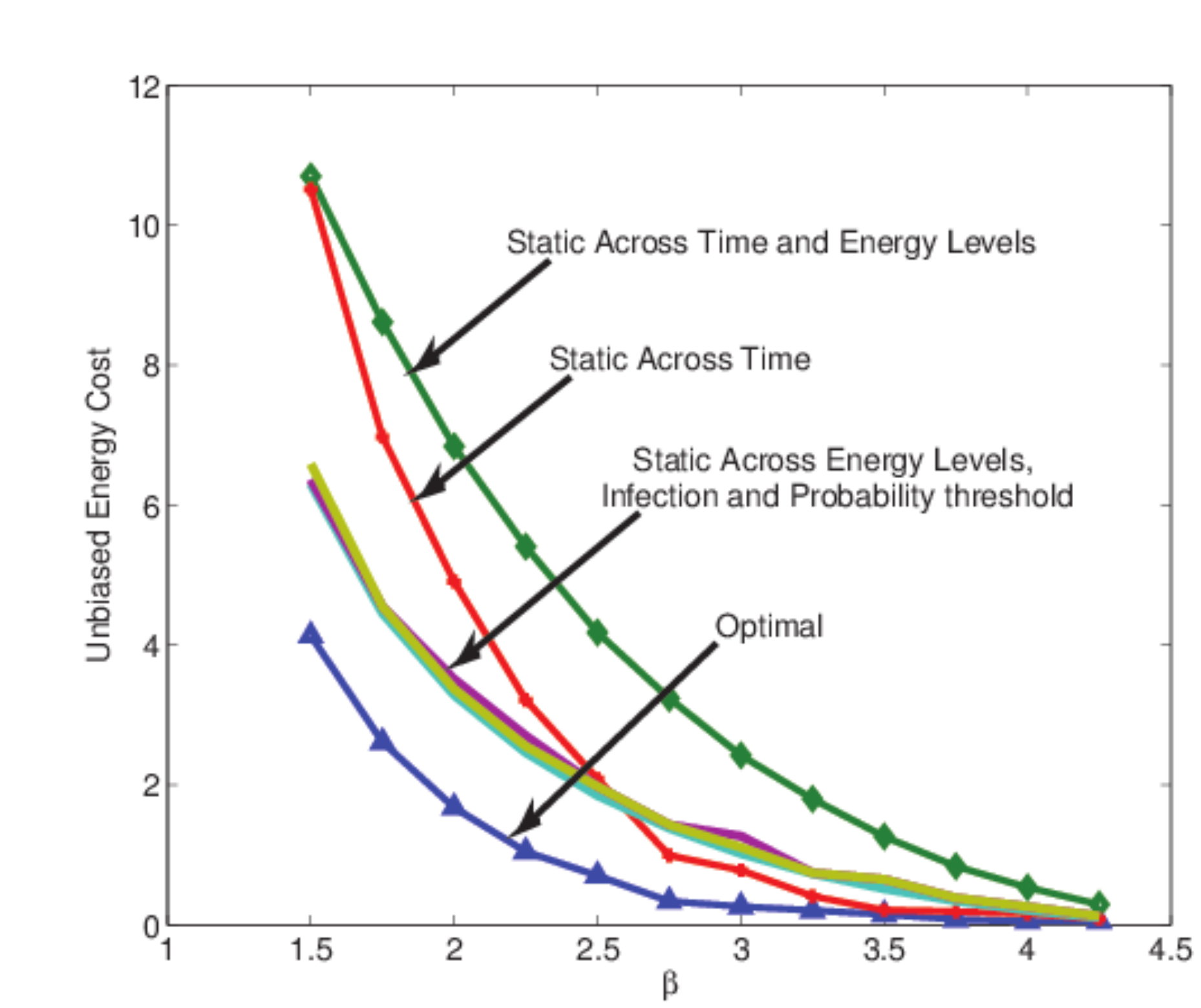}
\caption{{\scriptsize Performance of the optimal and  heuristic controls. The performances of the ``Static Across Energy Levels'', ``Infection Threshold'', and ``Probability Threshold'' policies are very close, and they are indicated with a single arrow.}}
\label{fig:perfor2}
\end{figure}
\subsection{Sensitivity of the optimal control to synchronization { and residual energy {determination}\hide{estimation} errors}}
\label{subsec:robustness}
{We will consider a system with $N=500$ nodes,  $\bI_0=(0,0,0,0.0125, 0.0125, 0.025)$, $T=5$, mandated probability of delivery  75\% and $\beta=2$ {simulated over 200 runs.}}
\hide{\subsubsection{Synchronization Errors}\label{subsec:synch}}
 \begin{figure}[htb]
 \centering
\includegraphics[scale=0.45]{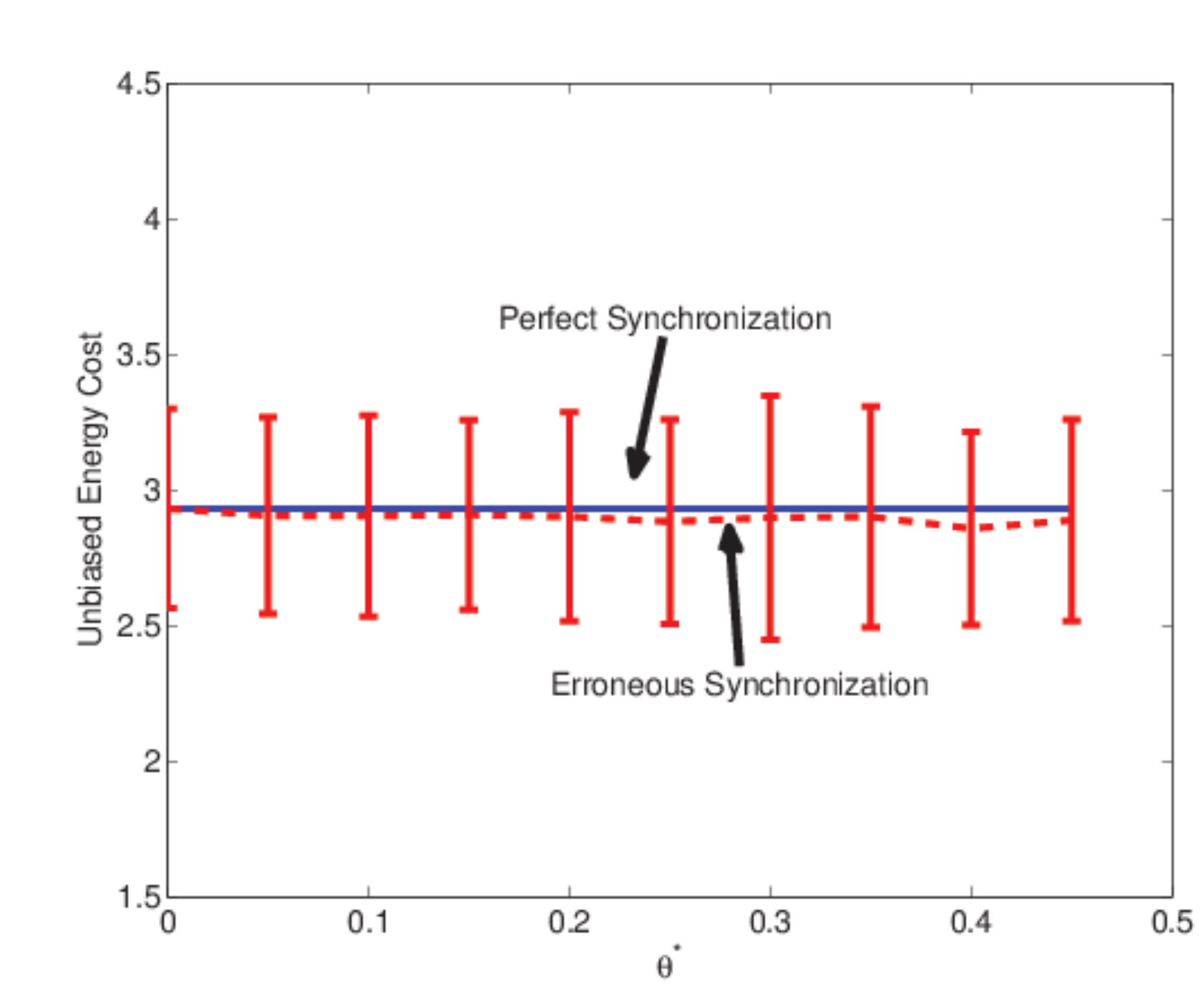}

\subfigure{{\scriptsize (a)  Unbiased energy cost}}

\includegraphics[scale=0.45]{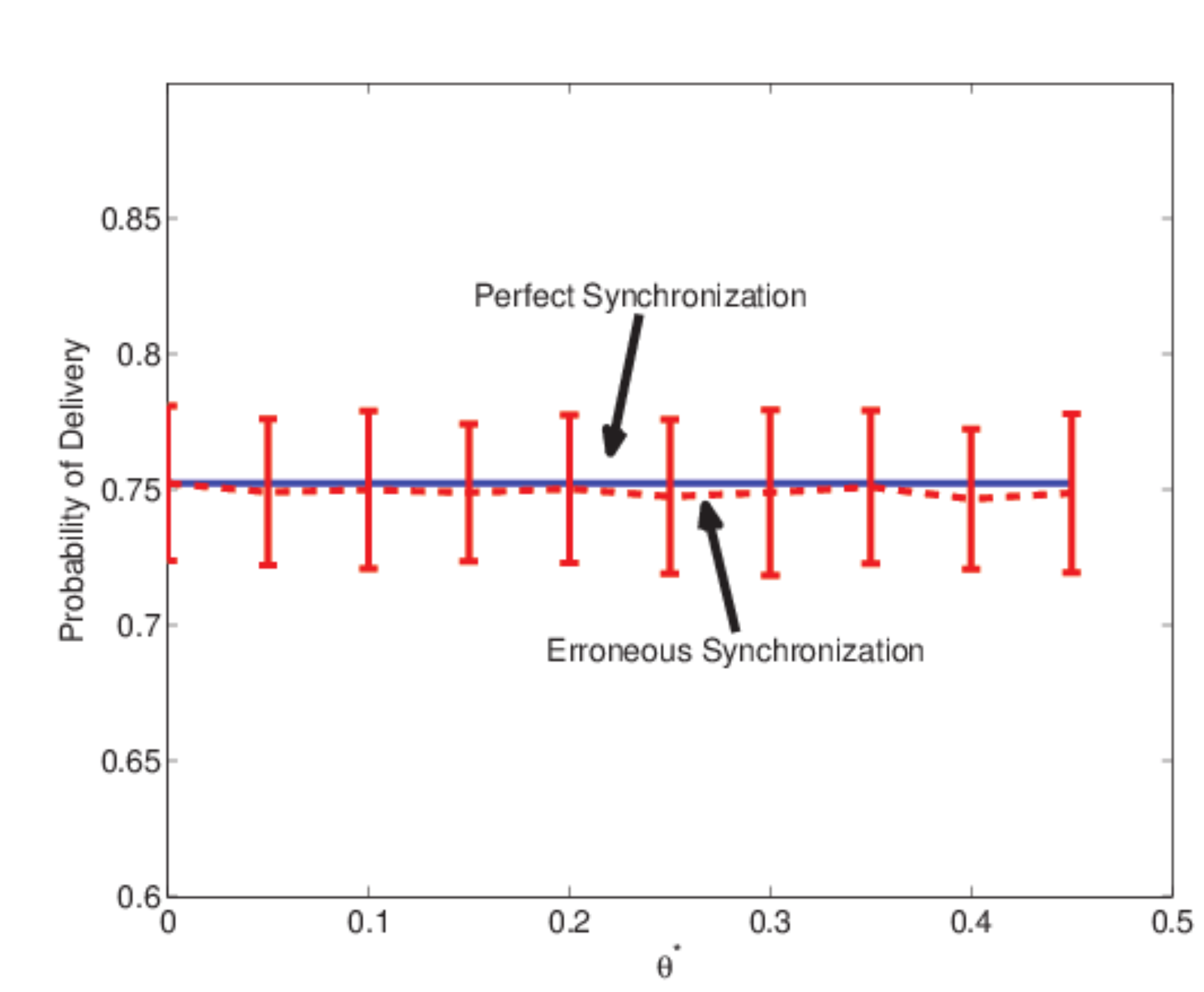}

\subfigure{{\scriptsize (b)  Probability of delivery}}
\caption{{\scriptsize {Comparison of the performance of the optimal policy when we have perfect synchronization (solid line) and an implementation with synchronization errors, in terms of both unbiased energy cost and probability of message delivery. $\theta^*$ is the range of the synchronization error for each node, and the error bars represent standard deviations.}
}}
\label{fig:drift}
\end{figure}
\subsubsection{Synchronization Errors}\label{subsec:synch}
 We allow each node to have  a clock synchronization error that manifests itself as a time-shift in implementing the control decisions.\footnote{In other words, if a node has  a time-shift of $\Delta$, while implementing the optimal control it uses a threshold time of $t_i + \Delta$ instead of $t_i$
when it has $i$ units of residual energy.} Thus, the optimal policy may incur a higher energy cost than the optimal value and provide a probability of delivery which is lower than the mandated value. We assess the extent of the deviations  considering node time-shifts  as mutually independent and uniformly distributed in $[-\theta^*,\theta^*]$; $\theta^*$ represents a measure of the magnitude of the synchronization errors. Fig.~\ref{fig:drift} reveals that  the network's performance is remarkably robust  in terms of both unbiased energy cost and probability of delivery (with maximum standard deviations of 0.5 for the unbiased energy cost and 0.03 for the probability of delivery) for   $\theta^*$ up to $10\%$ of the TTL $T$. This suggests that  the optimal policy does not have {a significant} operational drawback compared to the Static In Time heuristics that incur substantially higher energy costs except for large values of $\beta.$
{
\subsubsection{Energy Determination Errors}
Now we examine the case where each node is uncertain about its residual energy level, as may be the case for nodes with dated hardware. We assume each node under/over-estimates its residual energy level by one unit, each with probability $p^*$, independent of others. Specifically, if a node has $i$ units of energy, where $0 < i < B$, it estimates its energy availability as $i-1$, $i$ and $i+1$ with probabilities {$p^*$, $1-2p^*$ and $p^*$} respectively.\footnote{If a node has $B$ (respectively $0$) units of energy, it estimates its energy to be $B-1$ and $B$
(respectively $0$ and $1$) energy units with probabilities $p^*$ and $1-p^*$ (respectively $1-p^*$ and $p^*$).}  Fig.~\ref{fig:estvar} reveals that  the network's performance is  robust to such errors in terms of probability of message delivery, though the unbiased energy cost incurred increases slightly with $p^*$ (a change of less than $10\%$ for $p^*<0.15$). The maximum standard deviations of both cases are  similar to their analogs from \S\ref{subsec:synch}, confirming the previous observation. This suggests that  the optimal policy does not suffer  from any {significant} operational drawbacks as compared to the Static Across Energy Levels heuristics, which attain substantially higher energy costs.
 \begin{figure}[htb]
 \centering
\includegraphics[scale=0.45]{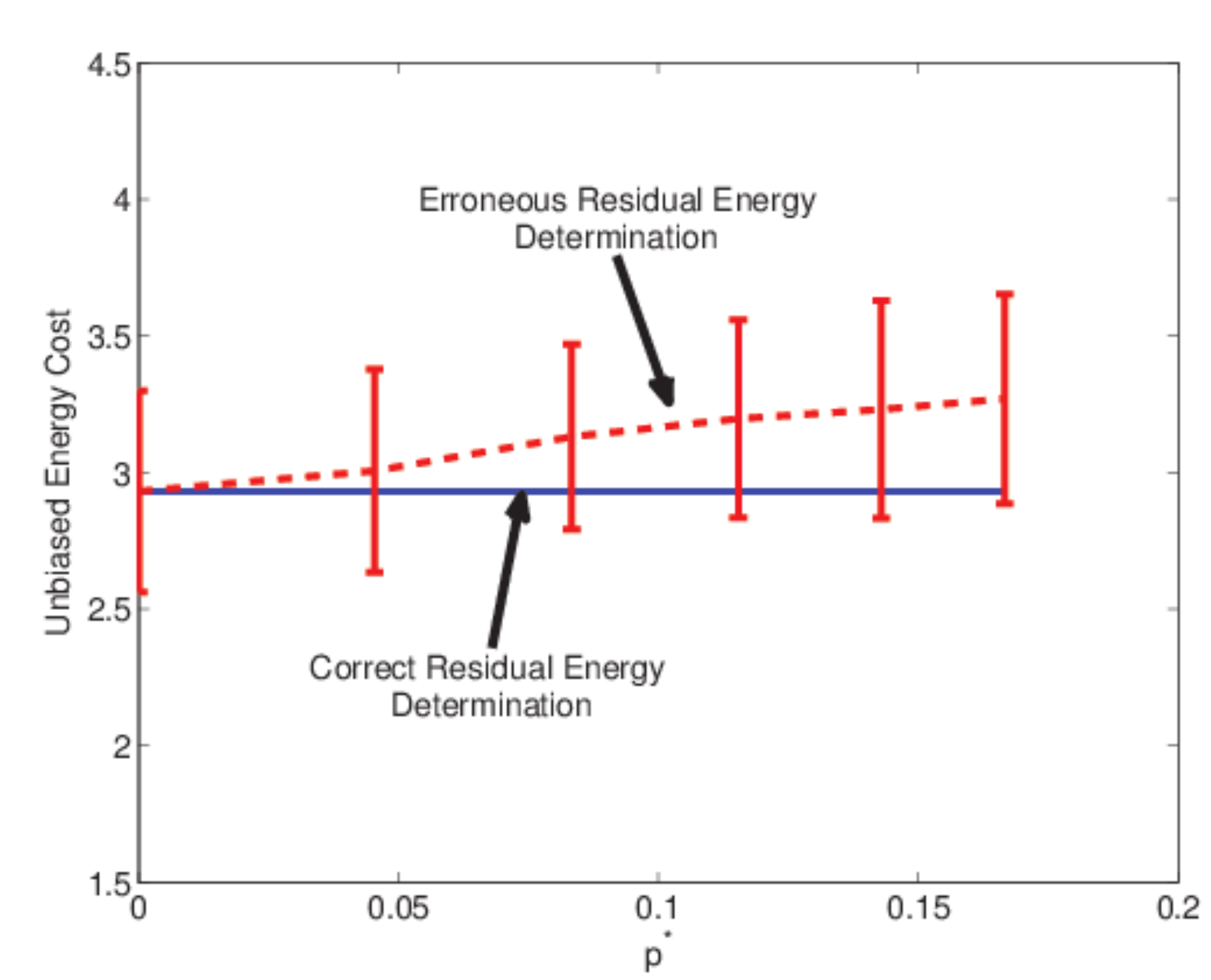}

\subfigure{{\scriptsize (a)  Unbiased energy cost}}

\includegraphics[scale=0.45]{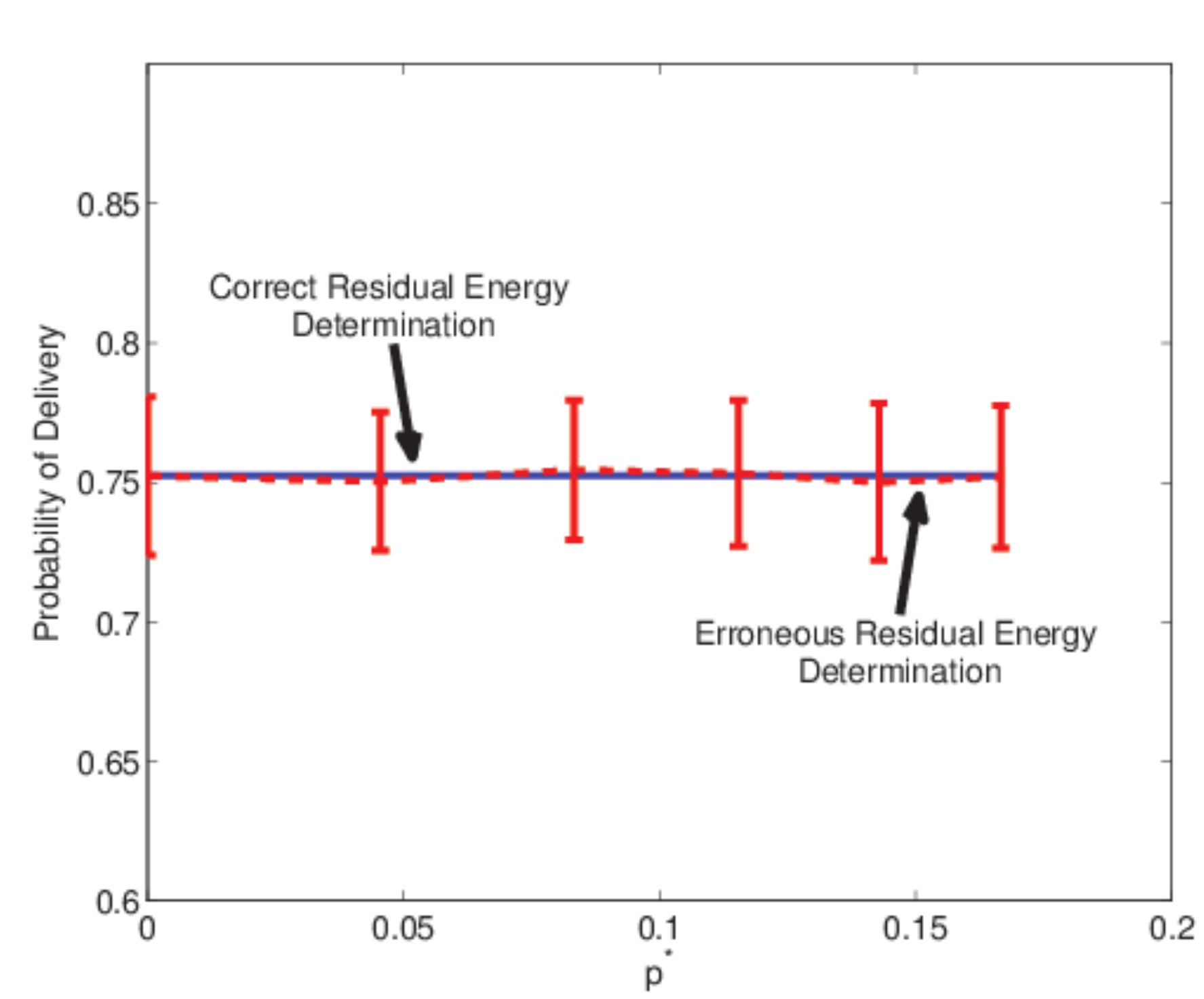}

\subfigure{{\scriptsize (b)  Probability of delivery}}
\caption{{\scriptsize {Comparison of the performance of the optimal policy when nodes have correct knowledge of their residual energy (solid line) with cases where each node can makes a one unit error in determining its residual energy level (with probability $p^*$)\hide{, in terms of both unbiased energy cost and probability of message delivery}.  Again, the error bars represent the standard deviations of each parameter.}
}}
\label{fig:estvar}
\end{figure}
}
\subsection{Multiple Message Transmission}\label{subsec:lifetime}
We now consider a scenario where the network seeks to successively transmit $M$ messages, where $M$ is a system parameter. Each message is associated with a TTL of $T$ and all nodes drop the message at the end of the TTL. The transmission of the $i$th message starts at the end of the TTL of the $(i-1)$-th message.
 The transmission of each message must satisfy the throughput requirement \eqref{path_constraint}.

We assume that at its initial time each message is uniformly spread to a fixed, say $\Upsilon$, fraction of  the  nodes that have at least $\taus + r$ units of energy. \hide{The spread is uniform among all energy levels
greater than or equal to $\taus + r$.  }Since each initial reception consumes $r$ units of energy,  the nodes that receive the initial copies of a message have enough (i.e., at least $\taus$ units of) energy  to subsequently forward the message after reception.\footnote{Here, $r + \taus = 3.$   Thus,  for example, if  $50\%$ of nodes have $4$ units of energy and $80\%$ of nodes have at least $3$ units of energy at the beginning of transmission of a message, and $\Upsilon = 0.01$,  then $1.25\%$ of the nodes with $4$ units of energy receive the initial copy of the message. So at the beginning of this transmission, $I_{3} = 0.00625$ and $S_4 = 0.04375.$  } \hide{This means that if for $k\geq1$, $\Upsilon_k:=(\upsilon^k_B,\ldots,\upsilon^k_0)$ is the spread of node energies {\em before} the $k$-th message (with the subscripts indicating energy level and the superscripts indicating the message), the message is spread among nodes in $\upsilon^k_i$, $i\geq\taus+r$. This causes the affected part of $\upsilon^k_{i}$ to turn into the initial infection $I^k_{i-r}((k-1)T)$, while the rest turns into $S^k_{i}((k-1)T)$. For $i<\taus+r$, $S^k_i((k-1)T)=\upsilon^k_i$. After the TTL of the message, all nodes drop the message, and for $\Upsilon_{k+1}$ (energy distribution before the next message) we have: $\upsilon^{k+1}_i:= S^k_i(kT)+ I^k_i(kT)$ for all $i$.\footnote{The fraction of initially infected nodes, i.e., $\sum_{i=\taus}^BI^k_i((k-1)T)$, is assumed to be fixed for all messages ($k\geq 1$).}   Because each  message consumes at least $r$ energy units from a fixed (non-zero) fraction of nodes, a finite number of messages  can be forwarded.} Once the network cannot guarantee the mandated probability of delivery for a message, we consider it to have been exhausted.

In these settings, we consider the natural generalization of our single transmission policies: the ``Myopic Optimal'' policy uses the one-step optimal for the transmission of each message, while others use the single-transmission best policy in their corresponding class (from \S\ref{subsec:heuristics}){\hide{for the transmission of each message.}}.  Our metric for comparing the performance of policies is the unbiased energy cost $\sum_{i=\taus}^B a_i \left(S_i(MT)+ I_i(MT)\right)-\sum_{i=\taus}^B a_i (S_i(0)+ I_i(0))$.
We only consider the cost of messages that can be forwarded to the destination before network exhaustion.

We plot the performance of each policy for $\{a_i\}$ that are quadratic functions of $B-i$ (Fig.~\ref{fig:diversity}), though similar results are seen for linear and exponential functions of $B-i$~\cite{techreport}. Here, $T=100$, $\beta=3, \Upsilon=0.001$, and the {mandated probability of delivery for each message} is 95\%.   Also, before the initial copies of the first message are distributed, all nodes have at least $3$ units of energy - $33\%$, $33\%$, and $34\%$ of the nodes have 3, 4, and 5 units of energy respectively. We see that the Myopic Optimal policy outperforms all the other policies that we consider in terms of the {unbiased} energy cost for each fixed number of transmissions and also the number of messages transmitted till exhaustion. Note that as $M$ increases, the difference between the unbiased energy costs of the Myopic Optimal policy and other policies  becomes substantial, e.g.,  the difference is around {$\hide{70-}90\%$} for around 10 transmissions. The number of messages forwarded to the destination till exhaustion by the Myopic Optimal policy is also slightly greater than that of the Static Across Time policy, and $60\%$ better than the best of the rest. 
 \begin{figure}[htb]
 \centering
\hide{\includegraphics[scale=0.45]{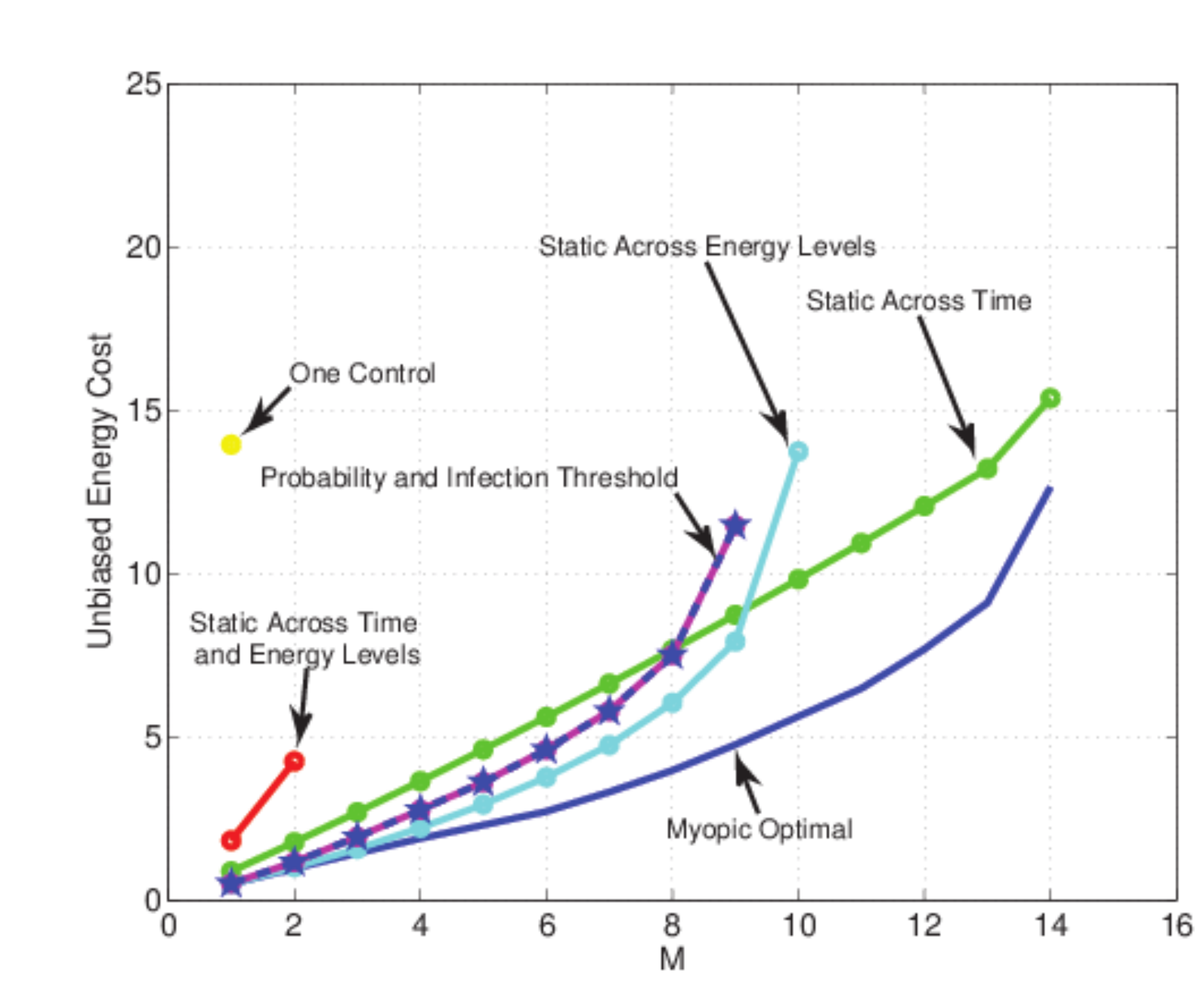}
\subfigure{{\scriptsize (a)  $a_i=e^{(B-i)}$}}}
\includegraphics[scale=0.45]{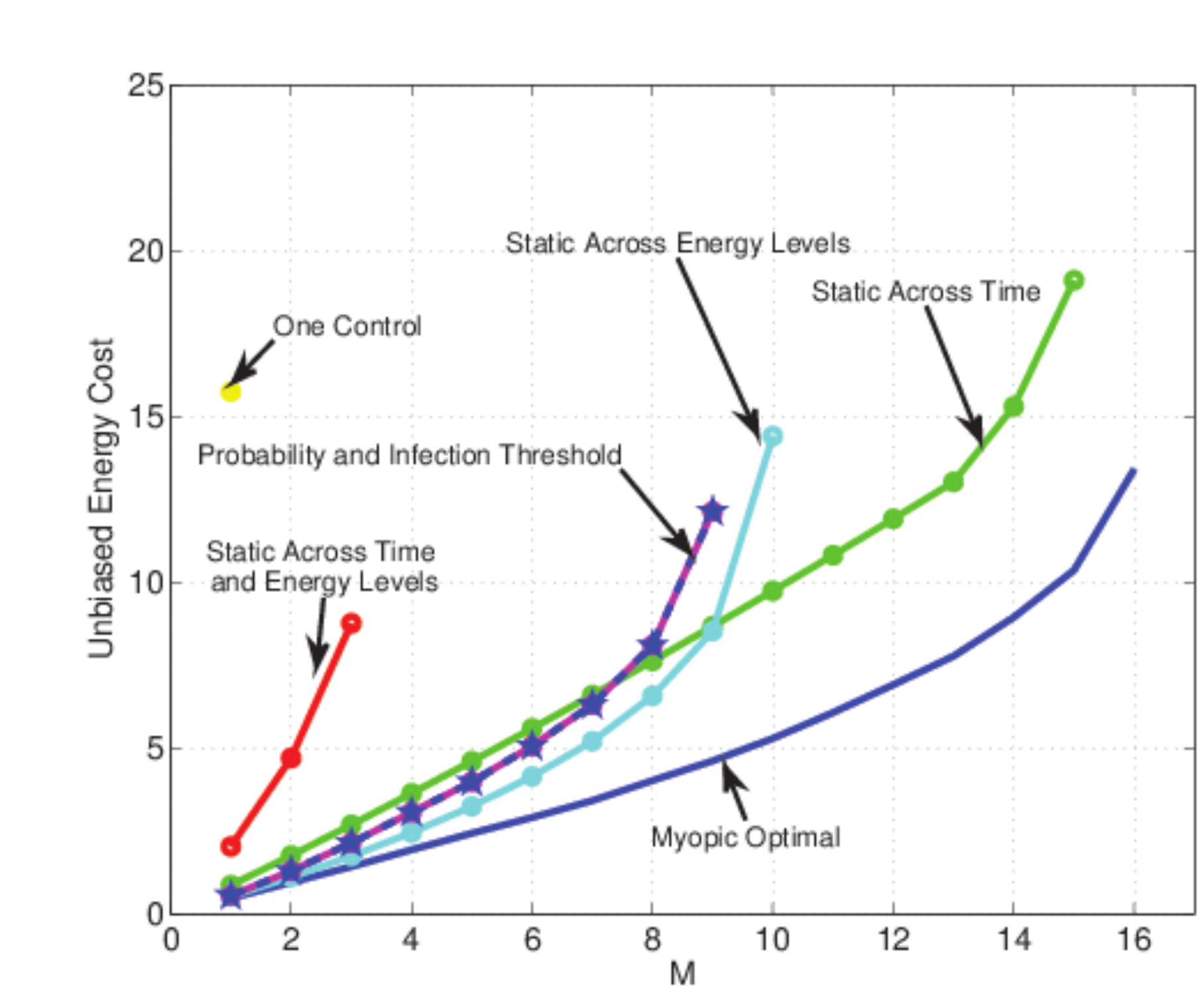}
\hide{\subfigure{{\scriptsize (b)  $a_i=(B-i)^2$}}}
\hide{\includegraphics[scale=0.45]{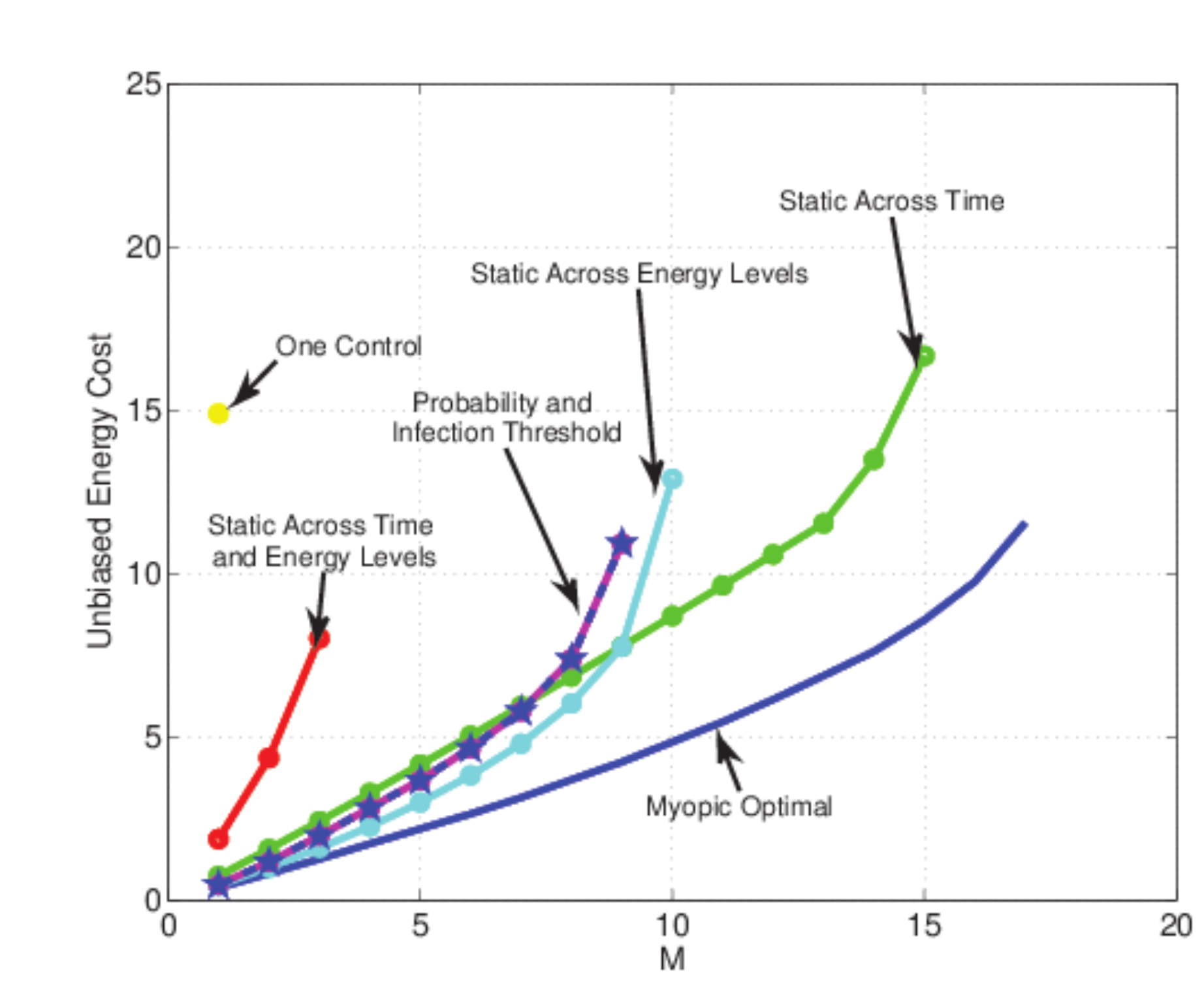}
\subfigure{{\scriptsize (c)  $a_i=(B-i)$}}}
\caption{{\scriptsize The figures plot the {unbiased energy cost\hide{ of the remaining energy at the delivery of a message}} as a function of {M, the number of messages transmitted,} for different policies. \hide{The parameters were:  $T=100$, $\beta=3$ and the mandated probability was 95\%, while the initial energy distribution was $(0,0,0,0.33,0.33,0.34)$ (before the distribution of the first message). }The battery penalties were
\hide{ (a) $a_i=e^{(B-i)}$, (b)} $a_i=(B-i)^2$\hide{, and (c) $a_i=(B-i)$ respectively}.\hide{, and each successive message was initially uniformly randomly spread to 0.1\% of the nodes that had enough energy to forward it.}
}}
\label{fig:diversity}
\end{figure}

\section{
Future Research}

\hide{We formulated  the problem of optimal energy-dependent message forwarding with (delay sensitive) throughput guarantees in energy-constrained DTNs as a multi-variable optimal control problem using a deterministic  mean-field model. We analytically established that optimal forwarding decisions are composed of simple threshold-based policies, where the threshold used by a node depends only on its residual energy.
We also proved that the thresholds are monotonic in the energy levels for a large class of cost functions associated with the energy consumed in transmission of the message.   Our simulations reveal that the optimal control substantially outperforms state-of-the-art heuristic strategies.

We now describe some directions for future research that have emerged from our investigation of the topic.}
Our analysis has targeted the transmission of a single message, and our simulations reveal that a natural generalization of the corresponding optimal policy substantially outperforms heuristics even for sequential transmission of multiple messages. It would be of interest to characterise the optimal policy in this case and also for the transmission of  multiple messages with overlapping time-to-live intervals. Next, in order to attain an adequate balance between tractability
and emulation of reality, we have abstracted certain features that arise in practice. A case in point is that we ignored the energy dissipated in scanning the media in search of new nodes. We have also assumed homogeneous mixing, i.e., the inter-contact times are similarly distributed for all pairs of nodes. Future research may be directed towards generalizing the analytical results  for models that relax the above assumptions.   For example, we may be able develop optimal policies for spatially inhomogeneous networks by partly relaxing  the homogenous mixing assumption using the approach of  \cite{2khouzani2012optimal}. Similarly, we have demonstrated using simulations that our optimal control policy is  robust  to clock synchronization errors and also errors in the determination of a node's residual energy level. Designing policies that are provably robust to the above errors as per some formal  robustness metric remains open.

\hide{Using numerical calculations and utilizing the benchmark of our optimal dynamic policy, we demonstrated that simple energy-based heuristic policies can achieve  close to optimal performance.

Our model can be generalized to incorporate the energy required for scanning the media, as well as energy replenishment, and the control can be made dependent on the energy-level of the receiver as well as the transmitter. We also intend to extend the results of this paper to the multicast scenario in which a message is destined for more than a single node and the objective is to maximize the number of destinations that receive a copy of the message.
As another future directions of research, we can relate the final state costs more \emph{directly} to the problem of energy-constrained DTN lifetime maximization. \hide{ We can also relax the deterministic model and investigate a partially observable Markov decision process (POMDP) model in general and characterize robust closed-loop optimal policies.}
}

\hide{\section{Acknowledgements}
The work is partially supported by the Army Research Office MURI Awards W911NF-08-1-0238 and W911NF-07-1-0376, and NSF grants CNS-0831919, CNS-0721434, CNS-1115547, CNS-0915697, CNS-0915203 and CNS-0914955.
}
\bibliographystyle{IEEEtran}

\bibliography{biblio}

\appendices

\section{Proof of Lemma~\ref{lem:5}}\label{appendix_lemma_5}
We only prove \eqref{eq:statement}; the proof for \eqref{eq:statement1} is exactly
the same and therefore omitted for brevity. We first establish:
\begin{Lemma}\label{lem:initialinterval}
If $\bar{\lambda}_0=1$, for each $j\geq \taus$  there exists a positive-length interval containing $T$ in
 which $u_j$ equals $0.$ In addition, irrespective of the value of $\bar{\lambda}_0$ and for all $t$, \begin{align}\label{eq:H(T)}
\ham(t^-)=\ham(T)= \lambda_E(T)\sum_{k=\taus}^B I_k(T).
\end{align}
 \end{Lemma}
\begin{proof} 
{Since the system is \emph{autonomous}\footnote{An autonomous optimal control {problem} is one whose dynamic differential equations and objective function do not \hide{have parameters which }explicitly vary with time $t$.}, the Hamiltonian is continuous in time and $\ham(t)=\ham(T)$ for all $t \in [0,T]$\hide{, as $\frac{\partial f}{\partial t}=0$ in this particular formulation of ~\eqref{problem} }\cite[p. 86 \& p. 197]{seierstad1987optimal}.}
We separately consider: $\bar{\lambda}_0=1$ and  $\bar{\lambda}_0=0$.

1) {\bf $\bar{\lambda}_0=1$.} The first part of the lemma clearly holds for $j\geq \taus$ if $\tau_j = T$, since then  $u_j(t) = 0$ for all $t \in [0, T]$.
We now seek to establish the same in the case that $\tau_j < T$, and therefore $I_j(T) > 0.$
At $t=T$, for $\taus\leq j\leq B$ we have:
\begin{align}\label{eq:varphi_T_0}
\hspace{-0.09in}\varphi_j(T)&=
\beta I_j(T)\bar{\lambda}_0
\sum_{k=r}^B\left(a_k-a_{k-r}-a_{j-\taus}+a_j\right)S_k(T). \hide{& &\taus\leq j\leq B}
\end{align}

Recall that $a_k$ is decreasing in $k$. Hence, since {${\bS}(0)\neq {\bf0}$}, and so for at least one $k\geq r$, $S_k(T)>0$ {(from Theorem~\ref{thm:constraints})}, for all $j\geq \taus$ we have $\varphi_j(T)<0$.%
\footnote{To see this, note that each term is negative as $a_{k-r}> a_k$ and $a_{j-\taus}> a_j$ for $k\geq r$ and $j\geq \taus$.}
Since $\varphi_j$ is a continuous function, $\varphi_j$ is negative in an interval of positive length including $T$. The first part of the lemma follows from \eqref{optimal_u_i:Gen}.

{Now, since $u_k(T)=0$ for all $k\geq \taus$ from the first part of this lemma,~\eqref{define:Hamiltonian:Gen} simplifies to \eqref{eq:H(T)}.}

2) {\bf $\bar{\lambda}_0=0$.}
Replacing $\bar{\lambda}_0=0$ in \eqref{eq:varphi_T_0}, it follows that $\varphi_j(T)=0$ for all $j\geq \taus$; the expression for the Hamiltonian in \eqref{hamsimpler} would thus lead again to \eqref{eq:H(T)}.

\end{proof}
{From \eqref{eq:co-states:Gen}, we have} $\dot{\lambda}_E=0$ , except at the points of discontinuity of $\bu$ -- a countable set -- leading to $\lambda_E(t)= \lambda_E(T)$ for all $t \in [0, T]$ due to the continuity of the co-states. Hence, from Lemma~\ref{lem:initialinterval},  the LHS in \eqref{eq:statement} becomes
\begin{align}
\lambda_E(T)\left(\sum_{j=\taus}^BI_j(T)-\sum_{j=r}^BS_j(t^-)-\sum_{j=\taus}^BI_j(t^-)\right).\label{eq:the_first_expr}
\end{align}
The lemma follows from two subsequently established facts:\\
(A) $\lambda_E(T)>0$,\footnote{
 In this part we show that $\lambda_E(T)>0$ whenever $\bu \not \equiv 0$. This combined with \eqref{co_st_finals:Gen} leads to $E(T)=-ln(1-p)/\beta_0$. Therefore, the delivery probability of the optimal control at the given terminal time $T$  equals  the mandated probability of delivery except possibly when $\bu \equiv 0$.}
and\\ (B)
$ \sum_{j=\taus}^BI_j(T)-\sum_{j=r}^BS_j(t^-)-\sum_{j=\taus}^BI_j(t^-)< 0$.

In order to establish (A), we rule out $\lambda_E(T)=0$; it must therefore be positive by  \eqref{co_st_finals:Gen}. We again consider two cases: (i) $\bar{\lambda}_0=0$ and (ii) $\bar{\lambda}_0=1.$ (i) If $\bar{\lambda}_0=0$, $\lambda_E(T)=0$ would lead to $(\bar{\lambda}_0,\vec{\lambda}(T),\vec{\rho}(T),\lambda_E(T))= \vec{0}$, which  contradicts {\eqref{vec_neq_zero}}. (ii) Otherwise (i.e., for $\bar{\lambda}_0=1$), let  $\lambda_E(T)=0$. Then,  $\lambda_E(t)=0$ for all $t \in [0,T]$. Thus, from \eqref{hamsimpler}, $\ham(t)=\sum_{j=r}^B \varphi_j(t) u_j(t)$.\hide{ \paragraph{Step 2: Proof of Lemma~\ref{lem:initialinterval}}}
Furthermore, since our system is autonomous and from Lemma~\ref{lem:initialinterval},  $\ham(t)=\ham(T)=0$ for all $t \in [0,T]$. \hide{\cite[p. 86 \& p. 197]{seierstad1987optimal}.} But, as argued after~(\ref{optimal_u_i:Gen}),  $\varphi_j(t) u_j(t) \geq 0$, for all $t \in [0, T]$ and all $j\geq \taus$.  Hence, we have $\varphi_j(t) u_j(t) = 0$ for all such $t$ and $j$\hide{$j\geq \taus$ and all $t \in [0, T]$}.
From Lemma~\ref{lem:initialinterval} and since $\bu \not\equiv 0$, there exists $t' \in (0, T)$ such that  $u_j(t)=0$ for all $t \in (t' ,T]$ and all $j\geq \taus$ and for some $k\geq \taus$, {\hide{there must exist}there exists} a non-zero value of $u_k$  in every left neighbourhood of $t'$.
 At any $t \in  (t', T]$  at which $\bu$ is continuous and from equations~(\ref{eq:co-states:Gen}), $
\dot{\rho}_j(t)=\dot{\lambda}_j(t)=0$ for $0 \leq j \leq B
$.
Since $\bu$ may be discontinuous only at a countable number of points and due to the continuity of the co-states,  $\rho_j(t')=\rho_j(T)=\lambda_j(T)=\lambda_j(t')=-
a_j$
for all $j\geq \taus$.

For $j\geq r$ and $k\geq\taus$, define $\Omega_{j,k}(t):=\lambda_j(t)-\rho_{j-r}(t)-\rho_{k-\taus}(t)+\rho_k(t)$. For all such $j, k$, we know that $\Omega_{j,k}(t')=
(-a_j+a_{j-r}+a_{k-\taus}-a_k)>0$. Hence, due to continuity of the co-states, there exists $\epsilon>0$ such that for all $t \in (t'-\epsilon, t')$ and all $j$, $k$, we have $\Omega_{j,k}(t)>0$.
But for all $t$, we had:
\begin{align*}
\ham(t)&=-\sum_{j=r}^B \bigg[\beta\sum_{k=\taus}^B \Omega_{j,k}(t) u_k(t)I_k(t)\bigg] S_j(t)\\&= -\sum_{r \leq j\leq B: S_j(0) > 0}  \bigg[\beta\sum_{k=\taus}^B \Omega_{j,k}(t) u_k(t)I_k(t)\bigg] S_j(t).
\end{align*}
The last equality follows since for each $j\geq 0$, $S_j(t) = 0$ at each $t \in (0, T]$ if  $S_j(0) = 0$ (Theorem~\ref{thm:constraints}).
Since $\bS(0)\neq{\bf 0}$ there exists $k\geq r$ such that $S_k(0) > 0$.
We examine a point $\bar{t}\in(t'-\epsilon, t')$ for which $u_l(\bar{t})>0$ for some $l\geq \taus$.
  Since $\ham(\bar{t})  = 0$, and every variable in the above summation is non-negative, $ \Omega_{k,l}(\bar{t}) u_l(\bar{t})I_l(\bar{t}) S_k(\bar{t}) = 0$.  Since $u_l(\bar{t}) > 0$, $I_l(\bar{t})>0$ by definition of $\bu$, and $\Omega_{k,l}(\bar{t})>0$, therefore   $S_k(\bar{t})=0$. This contradicts $S_k(0) > 0$  (Theorem~\ref{thm:constraints}). Thus, (A) holds.

We now seek to establish (B).  The proof follows from the key insight that it is not possible to convert all of the susceptibles to infectives in a finite time interval, and hence {at the terminal time the total fraction of infectives with sufficient energy reserves for transmitting the message is  less than the sum fraction of susceptibles and infectives with energy reserves greater than $r,\taus$ respectively at any time before $T$.}
To prove this, observe that for all $t\in[0,T]$, we have:
\begin{align*}
&(\sum_{j=\taus}^B \dot{I}_j+\sum_{j=r}^B\dot{S}_j)= -\beta \sum_{j=r}^B S_j \sum_{k=\taus}^{B}u_k I_k -\beta \sum_{j=\taus}^B u_j I_{j} \sum_{k=r}^{B}S_k\\&~~\qquad\qquad+ \beta \sum_{j=\taus}^{B-r} S_{j+r} \sum_{k=\taus}^{B}u_k I_k +\beta \sum_{j=\taus}^{B-\taus} u_{j+\taus} I_{j+\taus}\sum_{k=r}^{B}S_k\\&\qquad
\qquad=-\beta(\sum_{j=r}^{\taus+r-1} S_j \sum_{k=\taus}^{B}u_k I_k +\sum_{j=\taus}^{2\taus-1} u_j I_{j} \sum_{k=r}^{B}S_k) \leq 0.
\end{align*}
Thus $(\sum_{j=\taus}^B {I}_j+ \sum_{j=r}^B {S}_j)$ is a decreasing function of time, leading to
 $\sum_{j=\taus}^BI_j(T)- \sum_{j=\taus}^BI_j(t^-)-\sum_{j=r}^BS_j(t^-)\leq-\sum_{j=r}^BS_j(T).
$ 
   Now, since there exists $k\geq r$ such that $S_k(0) > 0$,  there will exist $k\geq r$ such that $S_k(T) > 0$ (Theorem~\ref{thm:constraints}). Also, from the same theorem, we have $S_m(T) \geq 0$ for all $m$. Thus, $\sum_{j=r}^BS_j(T) > 0.$ The result follows.


\end{document}